\documentclass[11pt,english]{article}
\usepackage[letterpaper]{geometry}
\usepackage{amsmath,amssymb,amsthm}
\usepackage{graphicx}
\usepackage{epstopdf}
\usepackage{setspace}
\usepackage{babel}
\usepackage{comment}

\theoremstyle{plain}
\newtheorem{theorem}{Theorem}
\newtheorem*{theorem*}{Theorem}
\newtheorem{lemma}{Lemma}
\newtheorem{corollary}{Corollary}
\newtheorem*{corollary*}{Corollary}
\newtheorem{definition}{Definition}
\newtheorem*{remark*}{Remark}

\def\mathscr{\mathfrak}

\def\a{\alpha}
\def\b{\beta}
\def\l{\lambda}
\def\Sw{\textrm{Sw}}
\def\Cut{\textrm{Cut}}

\makeatletter

\date{}

\makeatother

\begin{document}

\title{Optimal allocation in annual plants\\ with density dependent fitness}

\date{}

\author{Sergiy Koshkin, University of Houston-Downtown, Houston, TX, US\\ email: \texttt{koshkins@uhd.edu}\\   
Zachary Zalles, Rice University, Houston, TX, US,\\ 
Michael F. Tobin, University of Houston-Downtown, Houston, TX, US\\
Nicolas Toumbacaris, Emory University, Atlanta, GA, US\\
Cameron Spiess, Keene State College, Keene, NH, US
}

\maketitle

\vspace*{-3em}

\begin{abstract}
We study optimal two-sector (vegetative and reproductive) allocation models of annual plants in temporally variable environments, that incorporate effects of density dependent lifetime variability and juvenile mortality in a fitness function whose expected value is maximized. Only special cases of arithmetic and geometric mean maximizers have previously been considered in the literature, and we also allow a wider range of production functions with diminishing returns. The model predicts that the time of maturity is pushed to an earlier date as the correlation between individual lifetimes increases, and while optimal schedules are bang-bang at the extremes, the transition is mediated by schedules where vegetative growth is mixed with reproduction for a wide intermediate range. The mixed growth lasts longer when the production function is less concave allowing for better leveraging of plant size when generating seeds. Analytic estimates are obtained for the power means that interpolate between arithmetic and geometric mean and correspond to partially correlated lifetime distributions.
\bigskip

\bigskip


\end{abstract}

\newpage

\section{Introduction}

Models of allocation to growth and reproduction in plants based on the optimal control theory were pioneered by Cohen \cite{Coh71} (in the continuous case, by Denholm \cite{Den}), and developed in a variety of ways by many authors \cite{DeL,KR,MK,TY,VP,ZK}. Approaches and results are summarized in the survey paper \cite{PS}, and the book \cite{RB}. These models divide plants into two sectors, the vegetative (leaves, stems, roots) and the reproductive (flowers, fruits), and the model optimizes the allocation of net production between the two. Namely, it prescribes how to vary the fractions of net production allocated to the vegetative and to the reproductive growth over plant's lifetime to maximize a fitness measure.  Models with more than two sectors, that prescribe shoot and root allocation separately \cite{IR}, or add storage \cite{MK}, have been studied as well.

There are two sorts of assumptions about plants' lifetimes in allocation models of this type. In the more common deterministic models the lifetime is assumed to be fixed and the reproductive mass (assumed to be proportional to the seed yield) at the end of it is taken as the fitness measure. Two sector models then typically predict as optimal bang-bang behavior, which is allocation to only vegetative growth followed by allocation to only reproductive growth, with the switching time determined by the length of lifetime and the production rate \cite{Den,VP}. Models with a random distribution of lifetimes are more realistic but rare \cite{AmCoh,DeL,EKE,KR,LeCoh}, and typically assume that either geometric or arithmetic mean of seed yield is maximized. For geometric mean maximizers, a novel phenomenon is that the optimal behavior may no longer be bang-bang, but rather there is also a period of mixed growth when vegetative and reproductive organs grow simultaneously.

The main goal of this paper is to investigate effects of density dependence and temporal variation on the choice of fitness measure and the resulting optimal allocation schedules for annual plants. Density dependence is a simple example of eco-evolutionary feedback, where selection patterns in the population are affected by growing density of the population itself \cite{PP}. There is an ongoing discussion on what, if anything, is maximized by natural selection, especially when random fluctuations are taken into account \cite{ES,Frank,Roff}, and appropriate fitness measures often depend on the assumed type of feedback. Fisher's Malthusian parameter $r$ is a good fitness measure only for populations growing in a vacuum, with no appreciable density effects. MacArthur and Wilson termed selection under such conditions $r$-selection, and showed that in environments saturated with organisms $K$-selection takes place instead \cite{McAW}, where $K$ is the ``carrying capacity" of the environment, or, more precisely, the stable population size. For many species $K$-selection can be associated in deterministic models with maximizing the lifetime reproductive success $R_0$ \cite{ES}, which can be measured by seed yield for annual plants. The oft-used geometric and arithmetic mean maximization in fluctuating environments also correspond, roughly, to $r$ and $K$ selection, respectively.

However, already Pianka suggested that in real environments the two extremes are compromised giving rise to a continuum interpolating between $r$ and $K$-selectors \cite{Pian}. Moreover, if the environment is fluctuating then both types of selection might be contributing to the final outcome, and the environment itself might be affected by the selection taking place. More recent modeling suggests that fluctuating environments and density dependent feedbacks lead to transitions between fitness measures such as $r$ and $R_0$, and to maximizing mixed or hybrid measures that interpolate between the two \cite{AR17,AR20,ES,FS,LES}. Competition between different genotypes means that only relative difference in reproduction across competing genotypes matters, which implies diminishing returns on selective advantage from its increase \cite{Frank}. In agreement with the method of evolutionary stable strategies, such fitness measures reflect invasibility of a trait and/or the chance of its fixation in the long run \cite{FS,LES}.

Furthermore, the notions of $r$ and $K$ selection are linked to the classical logistic model, which has been criticized for including obscure phenomenological parameters and unrealistic predictions at the equilibrium. A recent model of nest site lottery, proposed in \cite{AB} and developed in \cite{AR17,AR20}, gives a deeper insight into actual demography and differences in selection for unsuppressed and suppressed growth, which stand behind the intuitions for $r$ and $K$ selection, respectively. It provides an example of a mechanism of selectively neutral growth suppression acting on the recruitment of juveniles (seeds and seedlings in the case of annual plants), and shows that transitional fitness measures should reflect more than just a difference in averaging. Juvenile mortality directly reduces reproductive value of seeds, and not necessarily proportionally. While availability of nest sites as a growth suppression factor is not applicable to plants, analogous mechanisms operate on seeds and seedlings as well. 

As we argue below (Section \ref{RepFit}), both mixed averaging and recruitment survival effects can be incorporated into a non-linear function of reproductive mass $L(y)$ that enters a fitness measure of the form $\mathbb{E}[L(y)])$, where $\mathbb{E}$ is the expected value averaged over individual lifetimes. We call $L(y)$ the {\it fitness function}. When the number of surviving seeds is proportional to the total reproductive mass, i.e. when recruitment survival can be neglected, one can take $L(y)=y$ for arithmetic and $L(y)=\ln y$ for geometric mean maximization, and those are the only two cases typically considered in the literature. But when non-linear effects of recruitment survival are taken into account, a much broader range of $L(y)$ is realized, from superlinear to sublogarithmic. 

The role of $L(y)$ is muted in deterministic models, because maximizing $y$ is equivalent to maximizing any monotone function of it, but its effect can be rather dramatic in allocation models with random variation, because $\mathbb{E}[L(y)]\neq L(\mathbb{E}[y])$. Mathematically, fitness functions are analogous to utility functions in economic models, although their interpretation is quite different from describing the subjective value of returns on investment. One can think of $L(y)$ as the ``evolutionary utility" of reproductive output $y$. In particular, linear and superlinear (convex) functions represent ``gambler's utilities" favoring high risk high reward strategies of delaying reproduction, while concave functions lead to more conservative behavior that Slatkin called ``hedging one's evolutionary bets" \cite{PS,SB,Slat}.

In addition to general fitness functions, we also consider general production functions assumed to be concave, instead of the typical linear ones, to reflect diminishing returns. In plant biology, concave production functions are natural as a consequence of Liebig's law of limiting factors \cite{Goud}. For example, self-shading of leaf canopies limits returns on increasing the leaf mass, other limiting factors include availability of minerals and nutrients. 

\subsection{Reproduction and fitness}\label{RepFit}

The fitness measure of our model takes into account both recruitment survival of seeds to seedling establishment, and temporal variability of lifetimes within a season. The two factors have to be taken into account separately because the allocation model applies only after plants begin to rely on autotrophy rather than on maternal resources for growth (as is typical). Since the model's lifetime distribution only covers established plants the seed yield as a measure of reproductive success has to be corrected for recruitment survival. In the long run, advantages and disadvantages of allocation schedules are sensitive not to the total number of seeds produced, but to the number of them that will live to reproduce, in the next or subsequent seasons. Indeed, the tradeoff between the risk of delayed reproduction and potentially larger seed yield may or may not be worthwhile depending on how many extra seeds will likely grow into reproducing adults. Even if we assume that the fraction of the overhead costs (flowers, pedicels and the like) is size independent and the total seed yield is proportional to the reproductive mass $y$, it is not the total yield that determines reproductive success, but rather the {\it effective} seed yield $S(y)$. The effective seed yield is the expected number of seeds that survive seed and seedling mortality, and are recruited into the adult population that may reproduce ($S$ stands for `surviving'). We model $S$ as a function of $y$ only, which is a simplification that can only be justified in special circumstances. In particular, it does not take into account that survivability may vary in the course of a season. 

When the effective seed yield is averaged over individual lifetimes to produce a fitness measure, an additional non-linearity may be introduced. This is familiar from the case of geometric mean, where the averaging transformation is $A_0(z)=\ln z$. More generally, mixed fitness measures correspond to power means with $A_\alpha(z)=z^\alpha$, where $0<\a\leq1$ is determined by the correlation between individual lifetimes in the population, as we show below. The overall fitness function is then the composition of the two non-linearities, i.e. $L(y)=A_\alpha(S(y))$, perhaps scaled and/or shifted by a constant for convenience. Thus, both recruitment survival and averaging contribute to the non-linear dependence of the fitness measure $\mathbb{E}[L(y)]$ on the reproductive mass.

\subsubsection{Averaging}\label{FitAv}

We start with the averaging as the simpler of the two factors. The conventional wisdom of geometric mean maximization goes back to a simple two allele haploid model of Dempster \cite{Dem}, which showed that the allele with the larger geometric mean gets fixed with greater probability. In a number of papers \cite{AmCoh, Coh66, Coh71}, Cohen and co-authors reached the same conclusion and explored its implications for the behavior of annual plants. More complex models confirmed the principle, but only under the assumption that the correlation of lifetimes between individuals of the same allelic type is near perfect (and cross-correlation among types is near zero) \cite{Frank}, e.g. because most of them die near the end of the season. Lewontin-Cohen \cite{LeCoh} and later Gillespie \cite{Gil}
also considered the other extreme, when the correlation is near zero, and found that the appropriate fitness measure is then the arithmetic rather than the geometric mean of seed yields. 

Quantitatively, the effect of temporal variability on the fitness measure was studied by Frank and Slatkin \cite{Frank, FS}, and more recently in \cite{LES}. In general, the genotype
with the greater probability of trait fixation has the larger value of $m-\rho\frac{\sigma^2}{2m}$, where $m$ is the mean, $\sigma^2$ the variance, and $\rho$ the correlation coefficient between lifetimes. For $\rho=0$ we get the arithmetic mean, and for $\rho=1$ a mean-variance approximation of the geometric mean. As $\rho$ varies between $0$ and $1$ the expression  $m-\rho\frac{\sigma^2}{2m}$ can be seen as the mean-variance approximation of power means $M_\alpha:=(\mathbb{E}[y^\alpha])^{\frac1\alpha}$ ($M_0:=\exp(\mathbb{E}[\ln y])$, the geometric mean), where $\mathbb{E}$ is the mathematical expectation. Indeed, $M_\alpha\approx m-(1-\alpha)\frac{\sigma^2}{2m}$, so $\alpha=1-\rho$. Thus, we can expect power mean maximization when plant lifetimes are less than perfectly correlated. 

If we assume that $S(y)$ is proportional to $y$, i.e. the same proportion of seeds survives regardless of their total number, one can take $L(y)=A_\alpha(y)$ because a constant coefficient of proportionality does not affect optimization. The corresponding fitness measure, which serves as the payoff functional in optimization, will then be of the form $\mathbb{E}[A_\alpha(y)]$, where $A_\alpha(y):=y^\alpha$ for $0<\a<1$ rather than $A_1(y):=y$ or $A_0(y):=\ln y$. 

According to Amir and Cohen \cite{AmCoh,Coh71}, ``linear" (i.e. arithmetic mean) maximizers are favored in large populations and stable environments, while ``logarithmic" (i.e. geometric mean) maximizers are favored in smaller populations and harsher environments. It is not hard to see the parallel between logarithmic and linear maximizers, and Pianka's $r$ and $K$-selectors, respectively. Transitional allocation patterns produced by environments that fall between this dichotomy have not, to our knowledge, been studied in the plant science literature. Biological or ecological processes reducing the correlation among lifetimes for a cohort of seeds will shift the selective advantage for allocation from geometric mean maximizers toward arithmetic mean maximizers, but not all the way. For example, some plants produce cohorts of seeds which germinate over multiple seasons due to delayed germination of a subset of the seeds \cite{Ven}. Seed dormancy for more than a year (seed bank) converts seasonal variability into cohort variability. Similarly, correlation among plant lifetimes could be reduced by spatial variability in resources due to the differential impact of resource availability on survival of spatially dispersed plants \cite{BLS}.  Disturbances like fire could increase or decrease the correlation depending on the spatial variability of fire induced mortality \cite{OWA}, thus shifting selection toward geometric or arithmetic maximizers, respectively.

\subsubsection{Recruitment survival}\label{FitRec}

As we saw, different averaging leads to a very narrow class of fitness functions. But it broadens considerably when non-linear effects of seed survival are also taken into account, i.e. when the averaging is applied not to $y$ directly but to $S(y)$. 

Ecological processes resulting in nonlinearity between seed yield and effective seed yield are very heterogeneous and complex. Among other things, seed survival is affected by seed dormancy, size, dispersal, physical or chemical defenses, and density dependent inhibition by other seeds, seedlings or adults, with these mechanisms selected for or against in concert with the allocation schedule. Determination of $S(y)$ is further complicated by the fact that seeds germinate in different years due to seed banks, and at different sites due to dispersal. Nonetheless, field studies of seed germination suggest \cite{Inou,Lin}
that seed and seedling density has some direct effects on germination, both positive and negative, that can be captured by $S(y)$. If the proportion of germinating seeds that successfully establish themselves decreases as seedling density increases we would expect a concave dependence of effective seed yield on total seed yield. 

More generally, effects of selectively neutral juvenile mortality can be captured by $S(y)$. Even simple known mechanisms that lead to mixed fitness measures, such as the nest site lottery \cite{AB,AR17} and its generalization, the host lottery that models parasitism \cite{AR20}, have close analogs for annual plants. For  example, a limited number of sites for establishment \cite{CPLO} would result in a fixed carrying capacity of the environment (in the mechanistic sense of \cite{Hui} rather than the phenomenological sense of the logistic growth model). The resulting function $S(y)$ will be linear for small $y$, but will saturate at the carrying capacity. A more flexible mechanism of the same nature is germination suppression to reduce competition when the seed or seedling density is high \cite{Inou,Lin}, which leads to more gradually saturating curves. Intraspecific competition among seedlings for resources, or herbivores and pathogens attracted to greater densities of host plants are also potential mechanisms producing a concave dependence \cite{Harp}.

Alternatively, a convex dependence can result from predator satiation. Predator satiation has been invoked to explain masting in perennial species \cite{KS}, but modeling suggests that this type of selective pressure could favor temporal clumping of seed production in annuals as well \cite{BV}.  Granivores often consume a lower proportion of seeds as the amount of seeds available increases, leading to the inverse density dependence and thus a convex shape of $S(y)$ \cite{HB}. Similarly, synchronized flowering or fruiting may improve efficiency of pollination and seed dispersal \cite{KS}. For example, larger seed crops were associated with more efficient dispersal by rodents in a forest system \cite{ZCX}.

As we saw, when $S(y)$ is linear one can take  $L(y)=y^\alpha$ for $0<\alpha\leq1$ depending on the type of averaging alone. On the other hand, one can expect initially convex dependence that saturates at larger values of $y$, such as given by functions of the form $S(y)=S_{\max}e^{-Cy^\alpha}$ with $\alpha<0$, from a combination of parasite/disperser suppression with a carrying capacity limit, see Figure \ref{Sexamples}(a). Composed with the geometric mean averaging $A_0(y)=\ln y$ this gives $L(y)=-Cy^\alpha+\ln S_{\max}$, where the multiplicative and the additive constants can be dropped for the purposes of optimization, leaving just $L(y)=-y^\alpha$. 
\begin{figure}[!ht]
\centering
(a)\includegraphics[width=1.8in,height=1.5in]{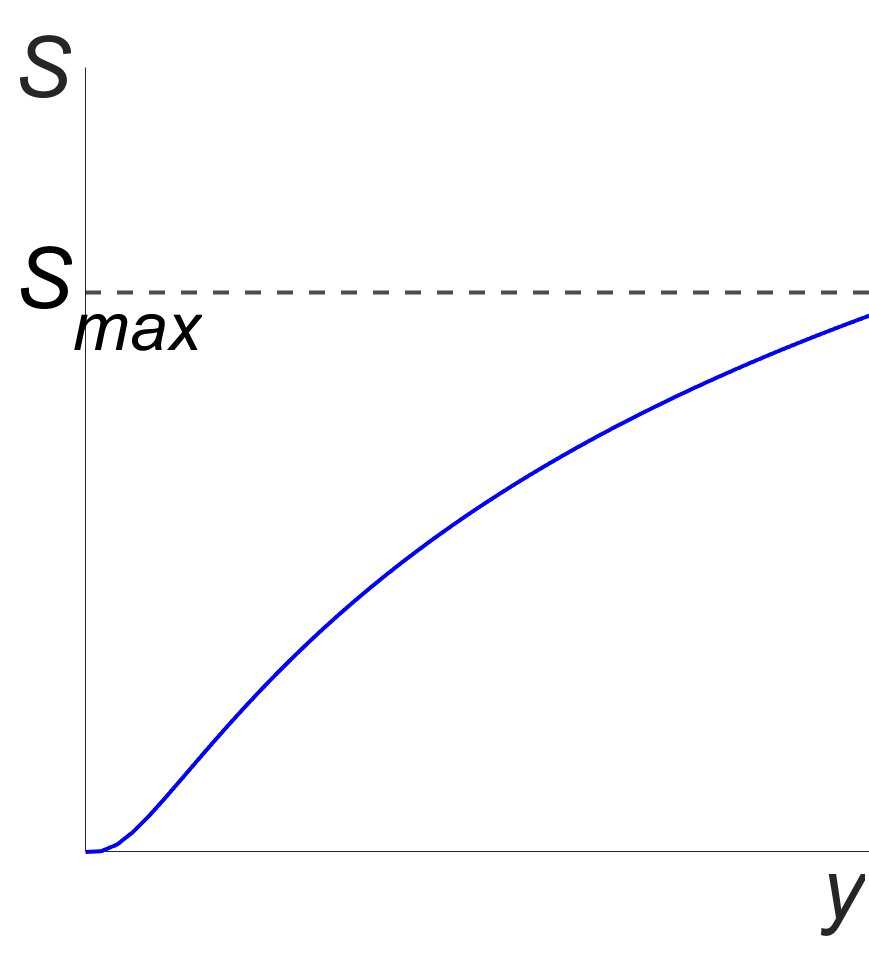}\hspace{5em}
(b)\includegraphics[width=1.9in,height=1.2in]{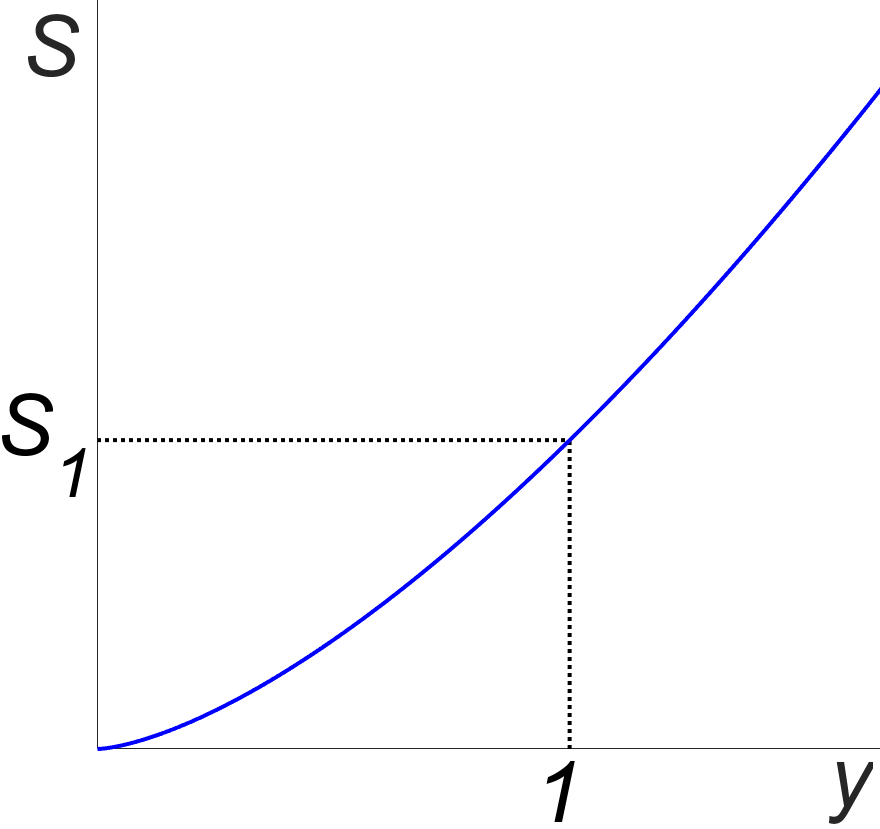}
\caption{\label{Sexamples} Possible graphs of effective seed yield as a function of reproductive mass: a) $S(y)=S_{\max}e^{-Cy^\alpha}$ with $\alpha<0$; b) $S(y)=S_1y^\gamma$ with $\gamma>1$.
} 
\end{figure} 

One can also consider $S(y)$ that remain convex over a feasible range of $y$ values, such as  $S(y)=S_1y^\gamma$ with $\gamma>1$, for a suitable normalization constant $S_1$, see Figure \ref{Sexamples}(b). Composed with the power averaging $A_{\alpha/\gamma}$ they produce the fitness function $L(y)=S_1^{\alpha/\gamma}y^\alpha$, where the constant can again be dropped. It is convenient to unify the $\pm y^\alpha$ cases by taking $L(y)=\frac{y^\alpha}{\alpha}$ for all real $\alpha\neq0$, which automatically takes care of the sign and simplifies some formulas.  We refer to $\alpha$ as the {\it fitness index}. Of course, much more general functions $S(y)$, and hence $L(y)$, are possible, and some of our general results apply to them. 

\subsection{Model description and main results}\label{Model}

The optimal allocation models we consider can be formulated in the standard optimal control setting \cite{LW,SS}. The state variables are $x$ (vegetative mass) and $y$ (reproductive mass), and the net rate of photosynthesis depends only on the vegetative mass given by the {\it production function} $P(x)$. The control variable $u$ is the fraction of the net production reinvested into vegetative growth, and the remainder $1-u$ is invested in reproduction.
The starting time $t=0$ is the time when the plants shift from maternal seed resources to photosynthesis for growth. We assume that plants' lifetimes are identically distributed, but not necessarily independent. Unfortunately, using something like a Gaussian distribution makes the problem analytically intractable. As a compromise, we split the lifetime into a ``safe" period, with fixed length $T_0$, and the ``volatile" period of maximal length $T$, which is equally likely to end at any time after $T_0$. The state equations are as follows:
\begin{gather}\label{StateEq}
\begin{cases}\dot{x} =uP(x),\ x(0)=x_0>0\\ \dot{y}=(1-u)P(x),\ y(0)=y_0\geq0;\end{cases}\\
0\leq u\leq1,\ \ \ \ \ t\in[0,T_0+T].\notag
\end{gather}
The payoff functional is the expected value of the fitness function $L(y)$ over possible lifetimes (we drop the $\frac1T$ factor that does not affect optimization):
\begin{equation}\label{ExtPay}
J[u](x_0,y_0,T_0,T)=\int_{T_0}^{T_0+T}L(y)dt\to\max\,.
\end{equation}
We generally assume that $P(x)$ is positive, and both $P(x)$ and $L(y)$ are twice differentiable and monotone increasing for positive $x$ and $y$. Other conditions will be explicitly stated when necessary. Under these conditions the existence of optimal control follows from the Filippov's theorem \cite[p.132]{SS}.  

To state the results more compactly we use a shorthand notation for different arcs of optimal allocation schedules. The letters below are V for vegetative, M for mixed, and R for reproductive growth.
\begin{definition}\label{ArcLet} We label an arc of an optimal schedule V when $u=1$, M when $0<u<1$, and R when $u=0$, respectively. A schedule is said to be of type VMR, etc., when the corresponding arcs are arranged from $t=0$ to $t=T$ in the indicated order left to right.
\end{definition}
\noindent In principle, a schedule can include any sequence of arc types in any order, but their variety turns out to be very limited for optimal schedules. The following theorem characterizes them under very general assumptions on production and fitness functions. 
\begin{theorem}\label{Schedules}
Let $\ln P$ be concave and $L$ be strictly monotone increasing. Then 

\noindent\textup{(i)} If $L$ is convex and $L(y)\xrightarrow[y\to\infty]{}\infty$ then any optimal schedule has at most one switch, and is bang-bang, i.e. R or VR;

\noindent\textup{(ii)} If $\ln L'$ is strictly convex then any optimal schedule has no more than three switches, and one of the following schedule types: R, VR, VMR, RMR, VRMR. Moreover, if $T_0$ is not too short (e.g. $T_0P'(x(T_0))\geq1$) and $y(0)=0$ then only VR and VRMR schedules are possible, and the first switch occurs during the safe period.
\end{theorem}
\noindent If $T_0$ is very short the plant may have to start reproducing from the start, which is not biologically feasible. The value of $P'(x(T_0))$ is typically easy to estimate from below by integrating $\dot{x}=P(x)$ on $[0,T_0]$. Since $y(0)=0$ means that the plant has no reproductive mass at the time of establishment, only VR and VRMR schedules have biological relevance when $L$ is as in Theorem \ref{Schedules}(ii). The first switching time $t_m>0$ is the time when the plant starts reproducing for the first time, and can be identified with the {\it age of maturity}. For realistic parameter values it is often close to $T_0$. When the schedule is not bang-bang, i.e. when it is VRMR, the first reproductive phase is often short and seems to be an artifact of our sharp separation between safe and volatile periods. With a more continuous distribution of lifetimes one can expect a direct transition from vegetative to mixed growth followed by purely reproductive growth.

To interpret the conditions on $L$ more concretely let us apply them to the power functions $L(y)=\frac{y^\alpha}{\alpha}$. They are convex for $\alpha\geq1$ corresponding to ``gambler's" behavior, and $\ln L'(y)=(\alpha-1)\ln y$ is strictly convex for $\alpha<1$. In other words, the possibility of mixed growth only exists for the more conservative indices $\alpha<1$. In fact, if the volatile period is long enough then mixed growth always occurs for some initial values. A quantitative threshold can be estimated in terms of two constants $q_c,\tau_c$ that depend only on $\alpha$, which we call the {\it cutoff constants}. They are easiest to interpret when the production function is linear with unit rate, $P(x)=x$. Then $\tau_c$ is the time-to-go at the end of the mixed period, i.e. (maximal) time remaining until the end of plant's lifetime, and $q_c$ is the ratio of reproductive to vegetative mass at that time. The constants are easy to find to any precision by solving a system of two non-linear equations, see Section \ref{CutPow}. For $\alpha=1, \pm\frac12$ and $-2$ they can be found analytically, and both grow to $\infty$ when $\a\to-\infty$. Some representative values are given in Table \ref{qctauc}, for $L(y)=\ln y$ they were found in \cite{KR}.
\begin{theorem}\label{Mexist}
Let $P(x)$ be concave and $L(y)=\frac{y^\alpha}{\alpha}$ with $\alpha<1$. If $T>\frac{q_c+\tau_c}{P'(x(T_0))}$ then any optimal schedule with $y(0)=0$ has a mixed growth arc. In particular, when $P(x)=kx^\beta$ it suffices that  $T>\frac{q_c+\tau_c}{k\b}\big(x_0^{1-\beta}+(1-\beta)kT_0\big)$.
\end{theorem}
\vspace{-1em}
\begin{table}[!ht]
\centering
\begin{tabular}{|c|c|c|c|c|c|}
\hline
 $\alpha$   & $1$ & 1/2 & $0$ & $-1/2$ & $-2$\\  \hline
$q_c$   & $0$ & 1/6 & $0.558$ &  1 & $1+\sqrt{2}$ \\ \hline
$\tau_c$   & $2$ & 5/2 & $2.793$ & 3 & $2+\sqrt{2}$ \\ \hline
\end{tabular}
\caption{\label{qctauc}Cutoff constants that determine occurrence of mixed growth.}
\end{table}
\noindent We should point out that unlike the condition on $T_0$, which is always satisfied in realistic circumstances, $T$ may well not be ``long enough". In fact, the minimal length explicitly depends on $\alpha$, and mixed phase  only occurs for a limited range of its values, other parameters being equal. 

Much more detailed results are obtained for linear production functions $P(x)=kx$, when the model can be solved (almost) analytically, see Section \ref{Lin}. The optimal growth trajectories are fully characterized by the evolution of mass ratio $q=\frac{y}{x}$, and the full picture of their dependence on initial values and model parameters can be visualized on a 2D diagram. This generalizes the results of \cite{KR} for $L(y)=\ln y$.

\subsection{Discussion}\label{Disc}

Since Cohen's work in 1960s, geometric mean maximization was associated with risk aversion and conservatism, while arithmetic mean maximization with more risk taking. It is easy to see why. In a large well-mixed population when lifetimes are uncorrelated we get a large number of essentially independent trials for the same allele type. Even if many plants die without offspring, survivors leave more than enough to compensate. On the other hand, when the correlation is perfect population size does not matter much because all plants die almost at once (e.g. due to water shortage or a cold spell). A vivid illustration of this point is that geometric mean becomes $0$ if even one of the terms is $0$. But in general, seasonal and individual adult mortality factors (like grazing by animals or microsite environment) combine and the correlation is neither perfect nor zero.

We argued that some non-linear effects of density dependence (such as lifetime variability and seedling mortality) can be represented by introducing a non-linear fitness function $L(y)$, whose mathematical expectation is maximized. Its role is analogous to the role of utility function in economics, but only linear and logarithmic functions were previously considered in the ecological literature. When reproductive mass is proportional to the effective seed yield  $L(y)=\frac{y^\alpha}{\alpha}$ was shown to be a natural choice, and the fitness index $\a$ being directly related to the correlation coefficient $\rho$ between lifetime distributions, $\a=1-\rho$. By varying $\a$ from $0$ to $1$ (correspondingly, $\rho$ from $1$ to $0$) our model traces changes in the optimal schedule. If juvenile mortality  saturates effective seed yields, negative values of $\a$ become meaningful as well, while seed predation and low disperser efficiency at low seed yields may lead to ``gambler's" fitness with $\a>1$. 

Two observations stand out. First, the time of maturity, when reproduction starts, grows with $\a$, moving from within the safe period for small values to within the volatile period for $\a\geq1$. Second, while on both ends optimal schedules are bang-bang, the transition is mediated by schedules with mixed growth for a wide range of intermediate values. On the upper end the mixed period disappears at $\a=1$ for any values of other parameters, and does not appear for ``gambler's" fitness $\a>1$, but on the lower end the threshold value depends on the initial plant size and the lifetime distribution. 
\begin{figure}[!ht]
\vspace{-0.1in}
\begin{centering}
(a)\includegraphics[height=25mm,width=40mm]{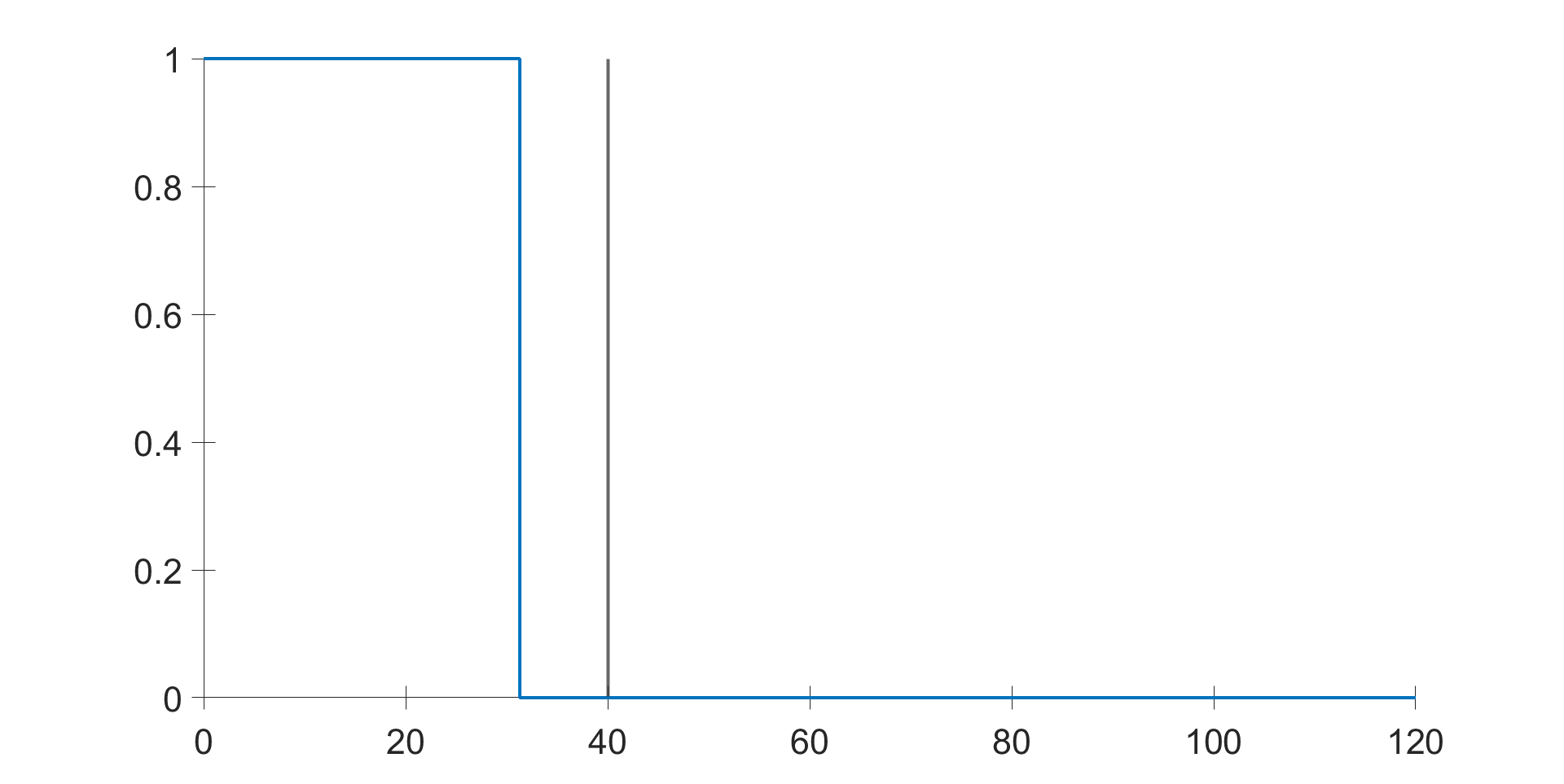}\hspace{0.2in} 
(b)\includegraphics[height=25mm,width=40mm]{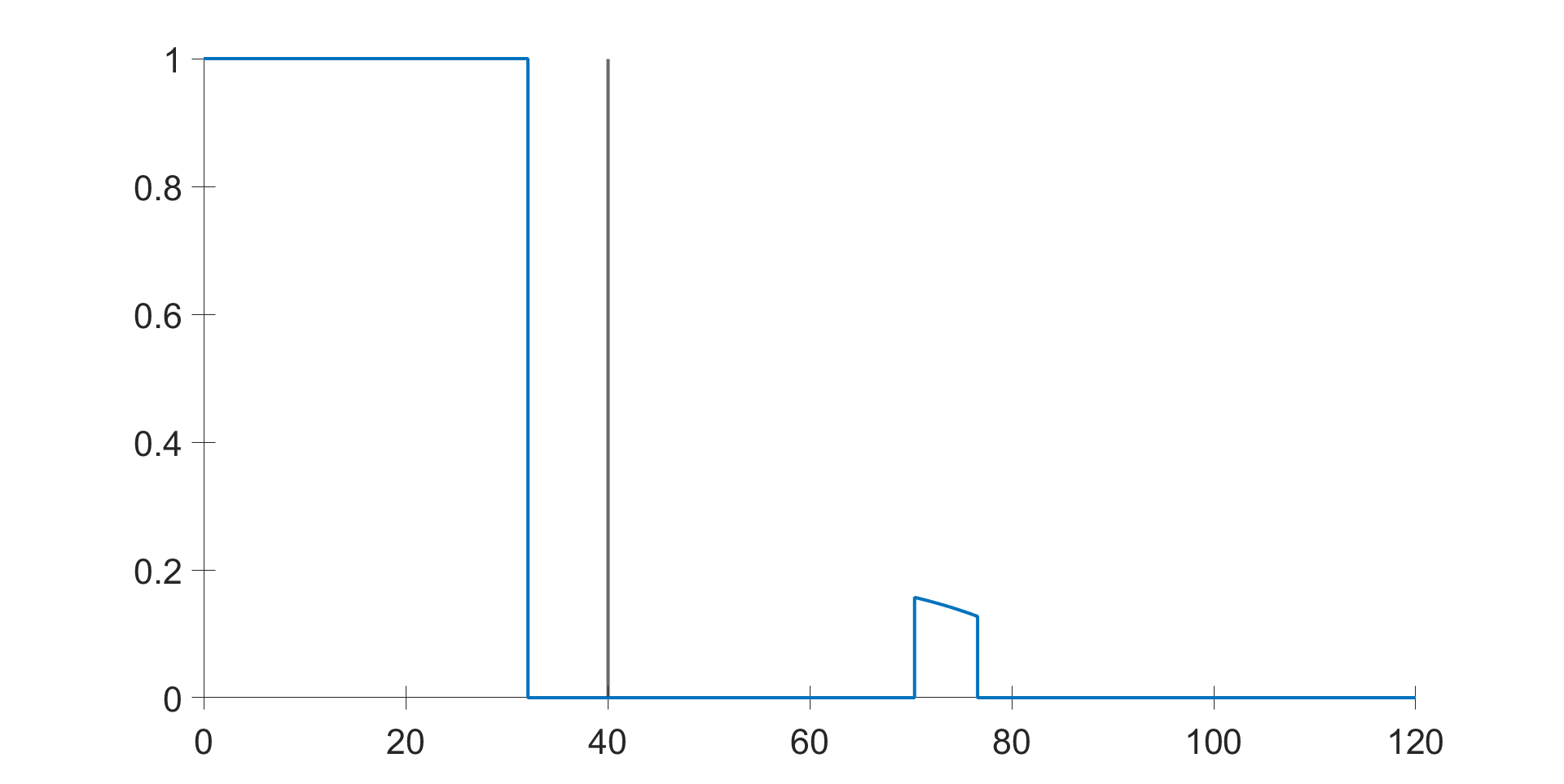} \hspace{0.2in} 
(c) \includegraphics[height=25mm,width=40mm]{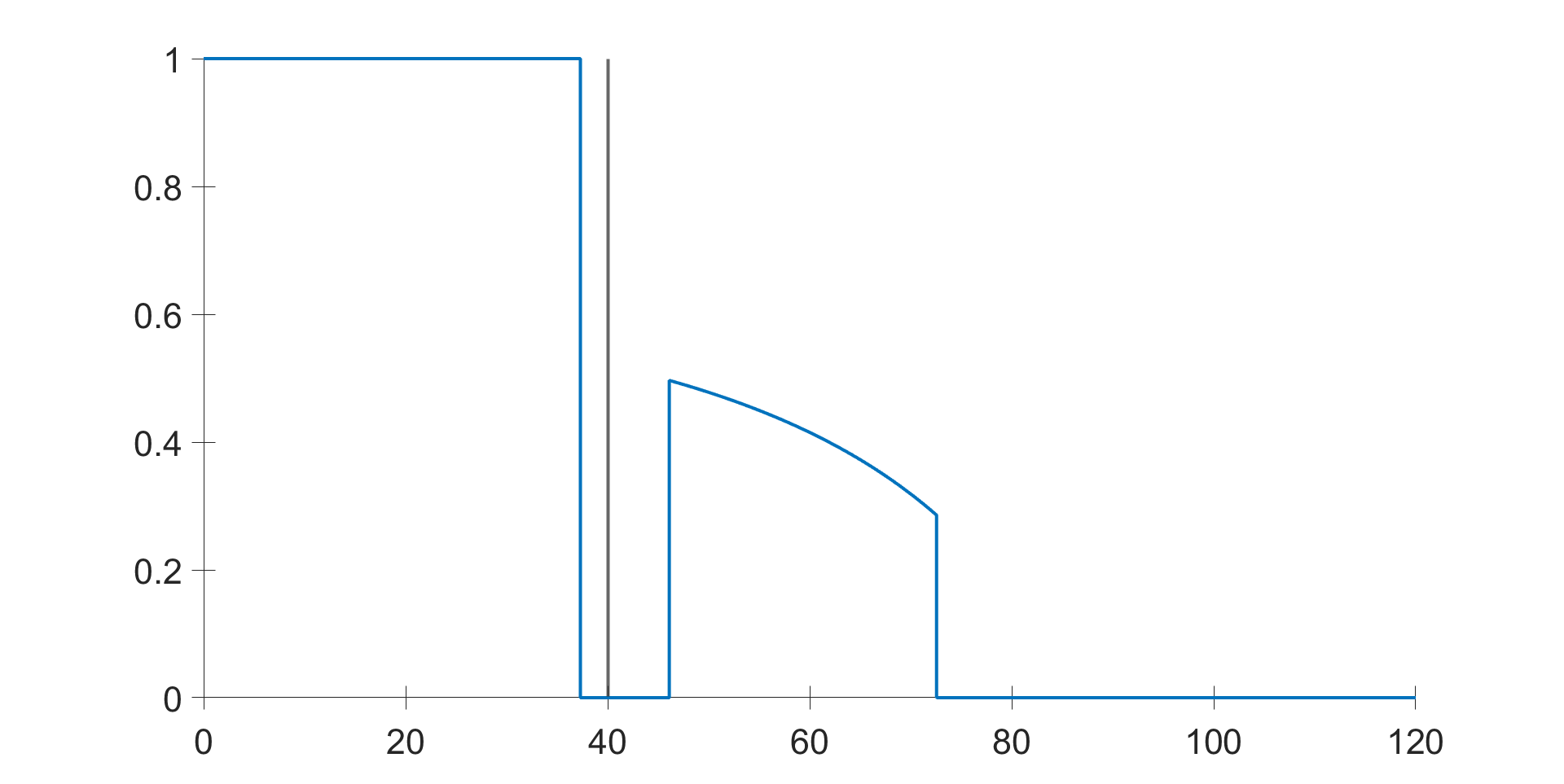}\\
(d)\includegraphics[height=25mm,width=40mm]{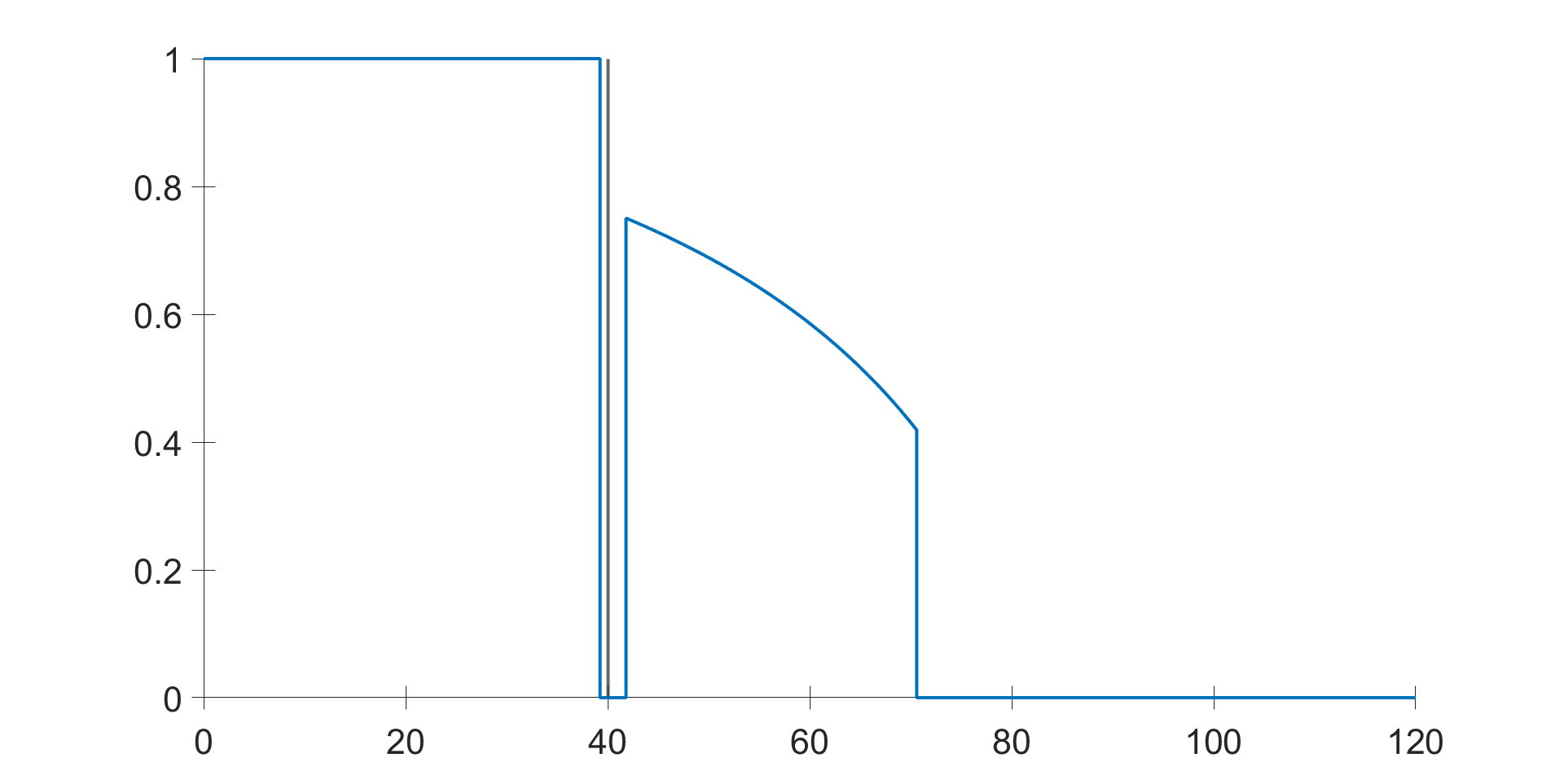}\hspace{0.2in} 
(e)\includegraphics[height=25mm,width=40mm]{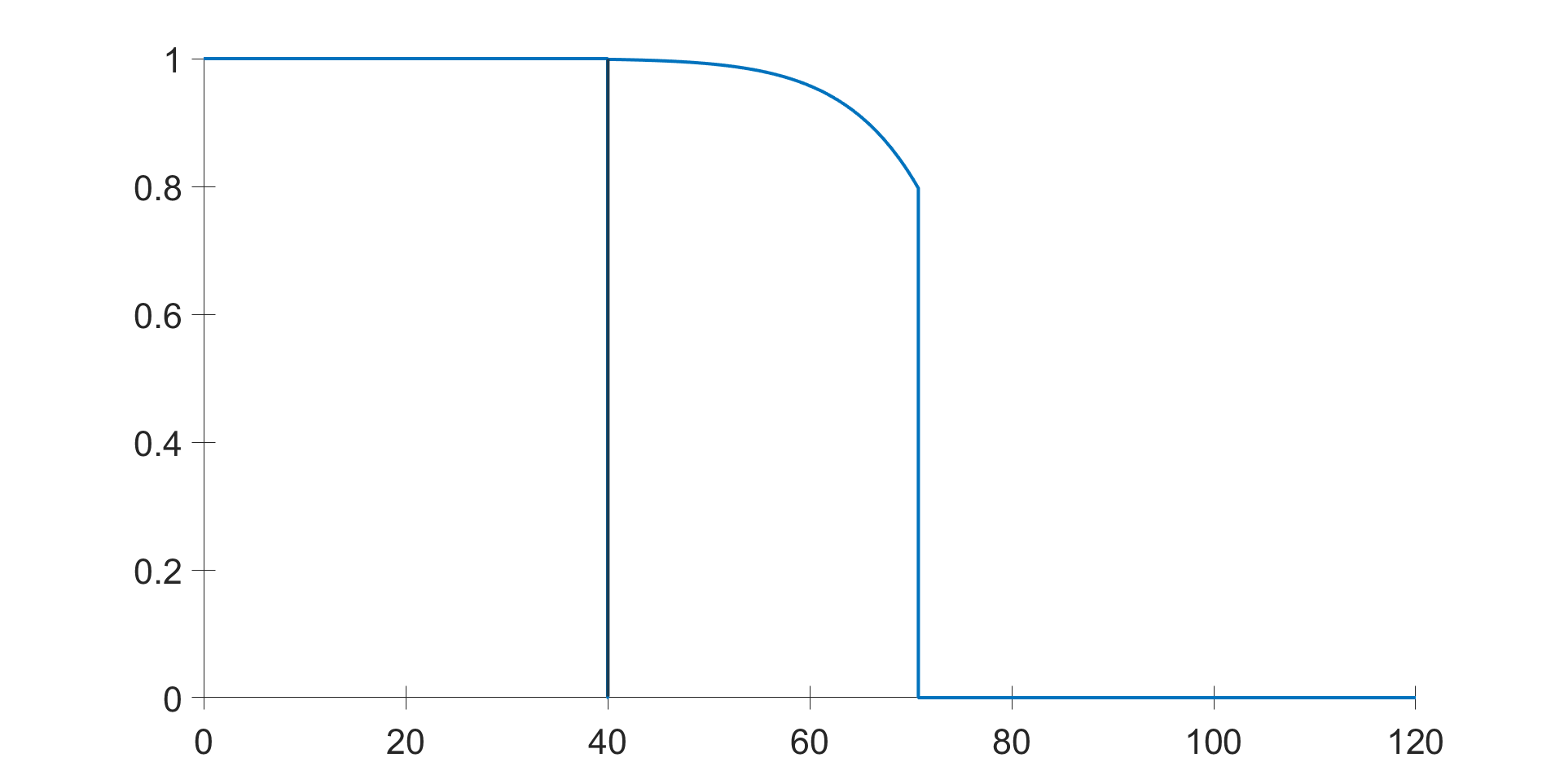} \hspace{0.2in} 
(f) \includegraphics[height=25mm,width=40mm]{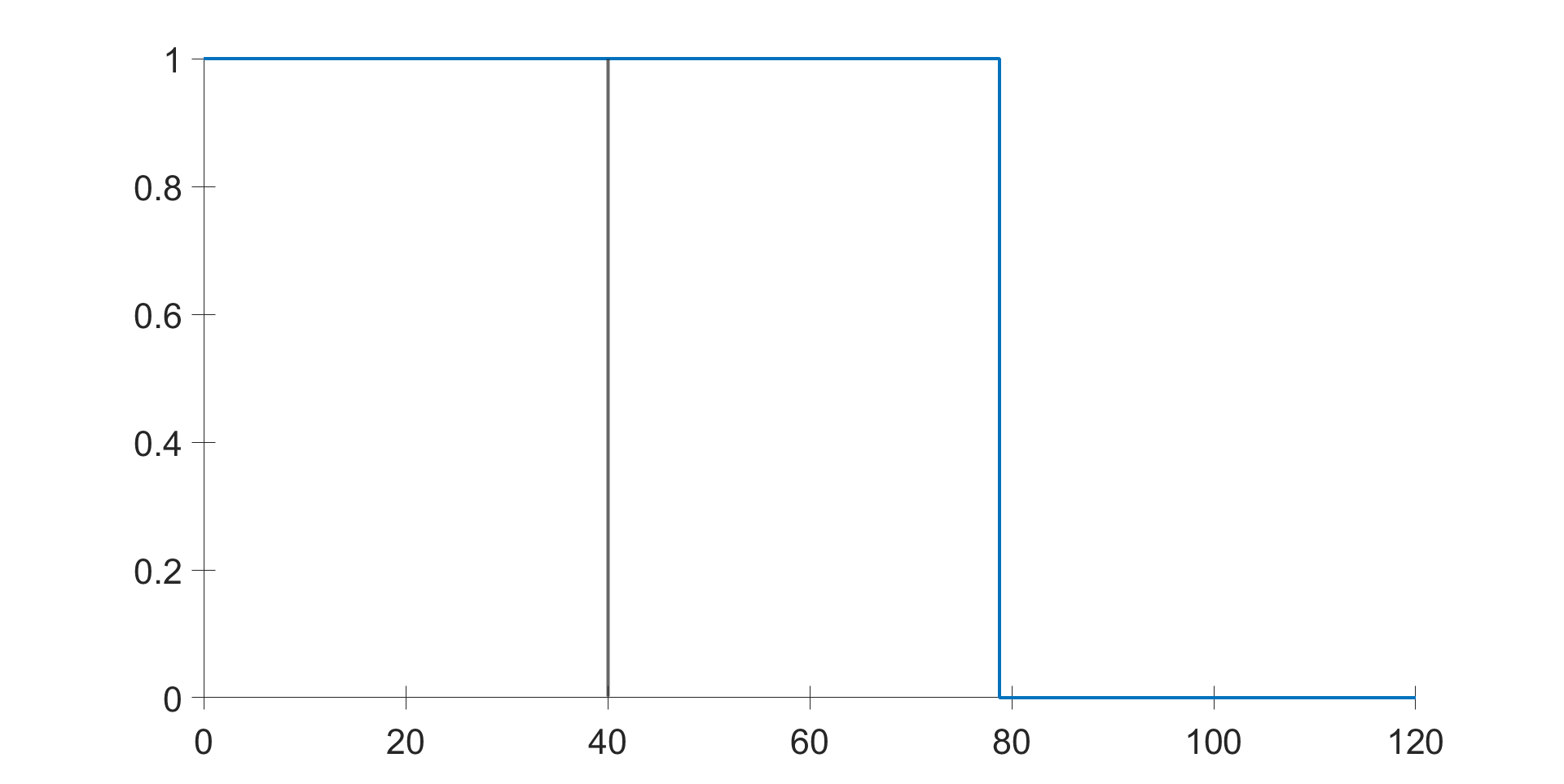}
\hspace{0.2in}
\par\end{centering}
\vspace{-0.2in}
 \hspace*{-0.1in}\caption{\label{OptU} Optimal schedules (graphs of $u(t)$) for $P(x)=0.1x^{0.8}$ with $T_0=40$ and $T=80$ days, $x_0=2.5$ mg, $y_0=0$, and (a) $\a=-4$ (b) $\a=-3$ (c) $\a=-0.5$ (d) $\a=0.1$ (e) $\a=0.7$ (f) $\a=1.5$.}
\end{figure}

Figure \ref{OptU} shows a progression of optimal schedules as $\a$ increases for a choice of parameter values. While thresholds and durations vary with model parameters the pattern is typical. The duration of mixed period increases with $\a$, and its nature shifts from reproductive growth dominating at smaller values to vegetative growth dominating as $\a$ approaches $1$. For the chosen parameter values it first appears at $\alpha\approx-3.5$, and at $\a=1$ it merges with the initial period of vegetative growth. No mixed growth occurs for $\a>1$, in accordance with Theorem \ref{Schedules}(i). In the course of mixed growth period the fraction of net production allocated to reproduction increases with time. The intermediate reproductive period, that precedes mixed growth, is barely noticeable for a range of $\a$, in this case $0<\a<1$, and is likely an artifact of the sharp separation between safe and volatile periods in our model. In this range the switch to mixed growth is very close to the beginning of the safe period.

For a continuous distribution of lifetimes one can still expect that the time of maturity moves from the lower tail of the lifetime distribution towards the average as $\a$ increases, and the initial vegetative growth transitions directly into mixed growth. It then follows a balanced path where the marginal values of vegetative and reproductive output are equal. The model also predicts that mixed growth is often followed by an opportunistic reproductive spurt at the end of plant's lifetime.

Mechanisms that support mixed growth in deterministic models include increasing the number of allocation sectors \cite{IR,MK}, introducing size-dependent mortality and/or fecundity \cite{JBMD,PS}, and size-dependent constraints on the size of reproductive organs \cite{IG,ZK}. As we saw, fluctuating environments lead to mixed growth being optimal under certain conditions even without any additional mechanisms. In continuous time models this phenomenon was first discovered by King and Roughgarden \cite{KR} for geometric mean maximizers and linear production. Our results extend the range where it is to be expected, and clarify limiting conditions on its occurrence. It is robust over a range of production functions with diminishing returns. It is ruled out for linear and convex fitness functions, as well as saturating ones that diminish selective advantage returns on reproduction too quickly. This last threshold is sensitive to the initial mass and lifetime distribution of plants in addition to the fitness index.

The structure of optimal trajectories $x(t),y(t)$ has a nice geometric description in terms of singular and switching surfaces that collect balanced growth paths and switching points between vegetative and reproductive growth periods, respectively. For linear production functions optimal dynamics depends only on the ratio of vegetative and reproductive masses, and can be presented on a two-dimensional diagram (Figures \ref{KRdiagram} and \ref{Extdiagram}), that visually manifests many features of the model predictions valid also in the general case. For $L(y)=\ln y$ such a diagram was presented already in \cite{KR}. Non-linear production shifts the onset of various stages depending on the initial values, but the overall pattern remains the same. As can be expected, mixed growth period lasts longer when the production function is made proportionally bigger, or less concave, because both provide greater opportunities for leveraging larger plant size when generating seeds.

In summary, our results suggest that when the length of lifetime fluctuates bang-bang behavior predicted by deterministic models cannot be taken as the baseline against which other mechanisms that induce mixed growth are assessed. The interaction of risk-averse bet-hedging with such mechanisms, as well as with seed bank effects, is of great interest for further study.
\bigskip

The rest of the paper is mostly mathematical and is organized as follows. In Sections \ref{Struct}-\ref{Lin} we study the model for the volatile period only, which is most challenging mathematically. In particular, in Section \ref{Struct} we introduce the marginal values of vegetative and reproductive masses, and characterize optimal schedule types, as needed to prove Theorem \ref{Schedules}. In Section \ref{Geom} we introduce geometric objects that determine which schedule type is selected for given initial values, and establish some of their properties. In Section \ref{Power} we specialize to power fitness functions, and prove the main ingredients for Theorem \ref{Mexist}. Section \ref{Lin} presents a detailed analytic solution for the case of linear production. Finally, in Section \ref{Ext} we incorporate the volatile period results into the full lifetime model.

\section*{List of symbols}

\begin{tabular}{c r c c r}
  $x,y$ & vegetative, reproductive mass, \ref{Model} &  & $\l,\mu$ & costate variables for $x$, $y$, \ref{Switch}\\
  $T_0,T$ & safe, volatile period, \ref{Model} & & $\xi$ & adjusted vegetative mass, \ref{Switch}\\
  $L(y)$ & fitness function, \ref{RepFit}& & $\tau$ & adjusted time-to-go, \ref{Switch} \\
  $P(x)$ & production function, \ref{Model}& & $q$, $r$ & adjusted mass ratios, \ref{Power}\\
  $\a$ & fitness index, \ref{RepFit}& & $q_c$, $r_c$, $\tau_c$ & cutoff constants, \ref{Power}; \ref{CutPow}\\
   $\Sw(\tau,\xi,y)$ & switching function, \ref{Switch}& & $t_s$ & VR switching time, \ref{Model}; \ref{Ext}\\
  $\Cut(\tau,\xi,y)$ & cutoff function, \ref{Switch}& & $t_m$ & time of maturity, \ref{Switch}\\
\end{tabular}

\section{Optimal schedule types for the volatile period}\label{Struct}

Most of the paper studies growth during the volatile period, which is more challenging mathematically, and the full lifetime model is considered in Section \ref{Ext}. Until then, we take the volatile period to start at $t=0$, but it is important to keep in mind that the plant is likely to be mature already at the starting point. 

In the following we use the standard notation and terminology of the optimal control theory \cite{LW,SS}. The payoff functional for the volatile period simplifies to:
\begin{equation}\label{Pay}
J[u](x_0,y_0,T)=\int_0^TL(y)dt\to\max\,.
\end{equation}
The costate (adjoint) variables are denoted $\lambda$ and $\mu$, and represent the marginal values of increasing $x$ and $y$, respectively. In other words, $\lambda(t)$ is the increase in the payoff from a unit increase in the vegetative mass at time $t$, and similarly for $\mu(t)$ \cite{LW}. The Hamiltonian of the model is
\begin{equation}\label{Ham}
H(x,y,\lambda,\mu,u)=L(y)+\big(\mu + (\lambda - \mu)u\big)P(x).
\end{equation}
The costate equations and the transversality conditions at the terminal time $T$ are given by Pontryagin's theorem:
\begin{equation}\label{Adj}
\begin{cases}
\dot{\lambda} = -\big(\mu + (\lambda - \mu)u\big)P'(x),\ \ \lambda(T)=0\\ 
\dot{\mu} = -L'(y),\ \ \mu(T)=0.\\
\end{cases}
\end{equation}
Since our system is autonomous, it satisfies the stationarity condition, i.e. the Hamiltonian remains constant along any extremal trajectory. The transversality conditions and \eqref{Ham} imply that this constant is equal to $L(y(T))$. Since the state equations and the payoff functional are linear in control we are dealing with a control-affine system, a class which is well known to feature both bang-bang and singular controls \cite{LW}. 

According to the Pontryagin maximum principle, the optimal control is obtained by maximizing the Hamiltonian as a function of $u$ at each time $t$. Since $\frac{\partial H}{\partial u}=(\lambda - \mu)P(x)$ does not depend on the control explicitly, the maximum of $H$ as a function of $u$ is attained at one of the endpoints, $0$ or $1$, as long as $\frac{\partial H}{\partial u}\neq0$. Since $P(x)>0$ for positive $x$ we have:
\begin{equation}\label{Bang}
    u(t)=\begin{cases}
        1,\ \lambda(t)> \mu(t)\\
        0,\ \lambda(t)<\mu(t)\,.
    \end{cases}
\end{equation}
Therefore, singular controls $0<u<1$ corresponding to mixed growth are only possible when $\lambda(t)=\mu(t)$ over an entire interval, in which case $\dot{\lambda}(t)=\dot{\mu}(t)$. It follows from the costate equations that $\lambda(t)$ and $\mu(t)$ are continuously differentiable as functions of $t$, as long as the control is piecewise continuous. Moreover, $\mu=\frac{L'(y)}{P'(x)}$ on any singular arc, and differentiating this identity while taking into account the state equations yields a feedback expression for the singular control: 
\begin{equation}\label{SingControl}
u = \frac{L''(y)+L'(y)\frac{P'(x)}{P(x)}}{L''(y)+L'(y)\frac{P''(x)}{P'(x)}}.
\end{equation}
While equations \eqref{StateEq} with this expression for singular $u$ can not usually be integrated explicitly, some useful information can be extracted from the stationarity condition. On a singular trajectory the Hamiltonian reduces to
\begin{equation}\label{SingHam}
H(x,y)=L(y)+\frac{P(x)}{P'(x)}L'(y),
\end{equation}
and remains constant, which gives an algebraic relation between $x$ and $y$ that holds during periods of mixed growth.

We will now determine possible switching structures of optimal control in the allocation problem \eqref{StateEq}-\eqref{Pay}.  Although {\it a priori} any sequence of arc types in any order is possible, it turns out that the possibilities in our case are quite limited. This is largely because the costate equations guarantee that both $\lambda(t)$ and $\mu(t)$ are strictly monotone decreasing. We summarize the result in the following theorem. Note that $\ln P$ is concave when $P$ is since $P$ and $P'$ are positive.
\begin{theorem}\label{5Structures} Let $P(x)$ and $L(y)$ be twice differentiable and $P(x)$ be positive for $x,y>0$. Then we have for the optimal trajectories of problem \eqref{StateEq}, \eqref{Pay}:

{\textup{(i)}} If $L(y)$ is strictly monotone increasing, i.e. $L'(y)>0$, then the final arc of an optimal trajectory is an R arc;
\smallskip 

{\textup{(ii)}} If $\ln P$ is concave, and an optimal trajectory has a V arc, this arc is initial;
\smallskip 

{\textup{(iii)}} If $\ln L'$ is strictly convex, and an optimal trajectory has an R arc preceding an M arc, this R arc is initial;
\smallskip 

{\textup{(iv)}} If $P(x)$ and $L(y)$ meet all the conditions of \textup{(i)-(iii)} then any optimal trajectory has no more than two switches, and one of the following control structures: R, MR, VR, RMR, VMR.
\end{theorem}
\begin{proof}
{\textup{(i)}} From the transversality conditions, $\mu(T) = \lambda(T) = 0$, and the adjoint equations \eqref{Adj}, we have that $\dot{\lambda}(T) = 0$, and $\dot{\mu} = -L'(y(T)) < 0$. Therefore, in a small left neighborhood of $T$ we have $\lambda(t)<\mu(t)$, and $u = 0$ by the maximum principle.
\smallskip 

{\textup{(ii)}} By contradiction, supposed that there is a switch to the V arc at time $t_s>0$. We have that $\lambda(t_s)$ =  $\mu(t_s)$ for the switching time $t_s$. As the arc before $t_s$ is V, we know that $\lambda$ must have been greater than $\mu$ for $t$ close to $t_s$, so that $\dot{\lambda}(t_s)\geq\dot{\mu}(t_s)$. 
Taking the derivative of $\eqref{Adj}$ when $u = 1$, we get $\ddot{\lambda}=(P'(x)^2-P''(x)P(x))\lambda$. Since $(\ln P)''=\frac{-P'(x)^2+P''(x)P(x)}{P^2(x)}<0$ we have $\lambda >0$, and $\ddot{\mu} = \frac{d}{dt}\big(-L'(y)\big)=0$, since $y$ is constant. Thus, if $\lambda > \mu$, $\mu$ and $\lambda$ cannot intersect again on $[0,t_s)$, as claimed. 
\begin{figure}[htbp]
\begin{centering}
\includegraphics[width=120mm]{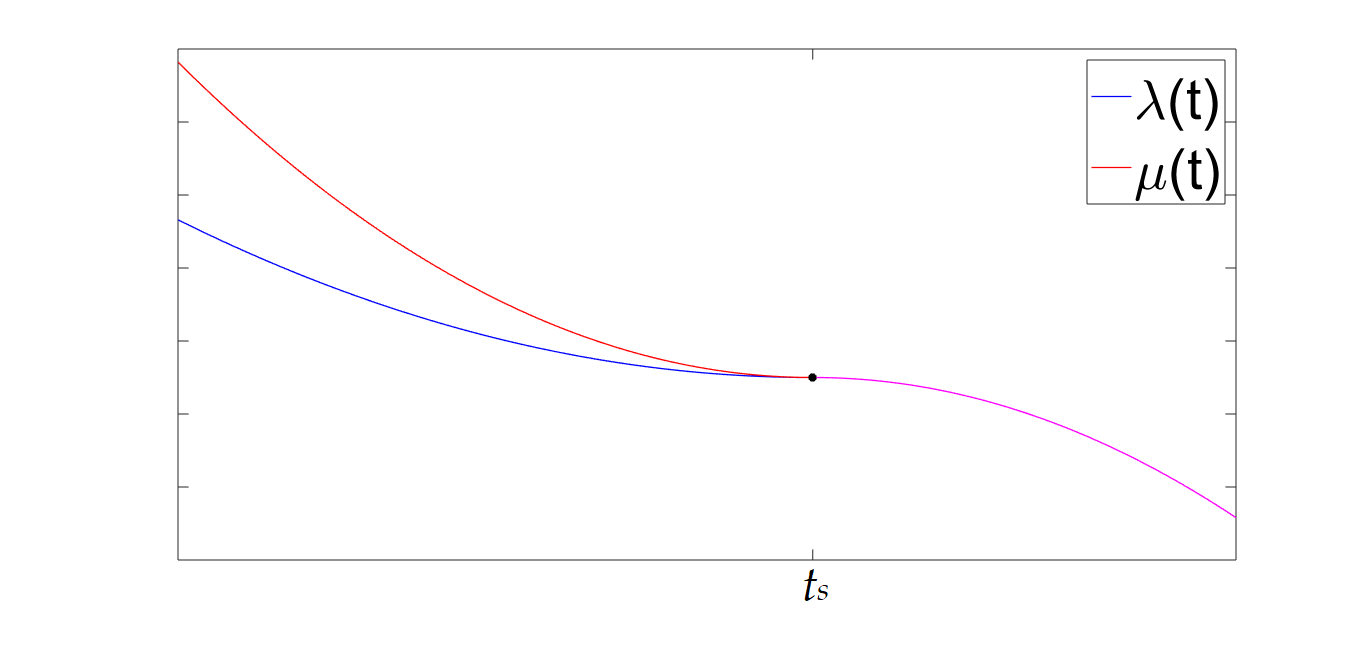}
\par\end{centering}
\vspace{-0.2in}
\hspace*{-0.1in}\caption{\label{Divlm} Marginal values on $[0, t]$ when $\mu(t)>\lambda(t)$ for $t<t_{s}$. The values are equal for $t\geq t_{s}$ (purple curve), so $t_s$ is an RM switching time.}
\end{figure}
\smallskip 

{\textup{(iii)}} Let $t_s0$ be the time of the RM switch. We have that $\lambda(t_s) = \mu(t_s) = \dot{\lambda}(t_s) = \dot{\mu}(t_s) = -L'(y_s)$.
As the arc before $t_s$ is R, we know that $\lambda < \mu$ on a small interval immediately before $t_s$. Furthermore, since $\lambda(t_s) = \mu(t_s)$ and $\dot{\lambda}(t_s) = \dot{\mu}(t_s)$, we know  $\ddot{\lambda}(t_s)\leq\ddot{\mu}(t_s)$. Solving $\eqref{Adj}$ with $u = 0$, we reduce this inequality to:  
$$
\frac{P'(x(t_s))}{P(x(t_s))}\leq-\frac{L''(y(t_s))}{L'(y(t_s))}<-\frac{L''(y(t))}{L'(y(t))}.
$$
The last inequality follows from $y(t)$ strictly increasing and $\frac{L''}{L'}$ strictly decreasing for $t<t_s$ since $\frac{L''}{L'} = (\ln(L'))'$. Therefore, $\ddot{\lambda}(t)\leq\ddot{\mu}(t)$, and hence $\lambda(t)< \mu(t)$ for $t<t_s$, see Figure \ref{Divlm}. Thus, there are no earlier switches, and the R arc is initial.
\smallskip 

{\textup{(iv)}} By (i), the last arc is always R. If there is no switching we have a pure R trajectory, and if it is preceded immediately by a V arc then the latter is initial by (ii), so we have a VR trajectory. If it is immediately preceded by an M arc, then either M is initial, and we have an MR trajectory, or the M arc is preceded by a V or R arc. Either way, those are initial by Lemmas (ii), (iii), respectively.
\end{proof}
\noindent Already for $P(x)=x$ and $L(y)=y^\a$ with any $\a<1$ there are values of $x_0,y_0$ and $T$ that instantiate each of the five structures (Section \ref{Lin}). On the other hand, for $\a\geq1$ only bang-bang, i.e. R and VR, structures occur (Corollary \ref{ConGrowth}). The next question is how to tell which structure is optimal for a given set of parameters.

\section{Geometry of optimal trajectories}\label{Geom}

It turns out that the switching structure of an optimal schedule is determined by the location of the initial point relative to a geometric configuration, two surfaces and the curve of their intersection, see Figure \ref{SwSing}. We describe the elements of this configuration in this section, and sketch how it determines optimal strategies at the end of \ref{Sing}. For linear production functions the configuration collapses to a two-dimensional one described in Section \ref{Lin}.

\subsection{Switching surface}\label{Switch}

Let us start by looking at the VR trajectories. The marginal values are both equal (to zero) at the terminal time $T$. For the final R arc to join to the V arc at an earlier time, they must become equal again at that time. This motivates the following definition.
\begin{definition} We say that a point $(x,y,t)$ is on the switching surface if $t<T$, and the costate variables computed backwards in time along a pure R extremal trajectory starting at it, are equal.
\end{definition}
\noindent The trajectory in the definition is uniquely determined. Denote $t_s$ the switching time, and $x_{s},y_{s}$ the values at the switch. Since $u=0$ after the switch we have that $x(t)=x_s$ is constant, and $y(t)=y_s+P(x_s)(t-t_s)$. It is convenient to use $y$ as the integration variable when integrating the costate equations \eqref{Adj} backwards from $T$ to $t_s$. Denoting $y_T=y(T)$, we obtain:
\begin{equation}\label{AdjR}
\mu(y) = \frac{L(y_T) - L(y)}{P(x_s)},\quad
\lambda(y) = \frac{P'(x_s)}{P(x_s)^2}\int^{y_T}_y L(y_T) - L(y)\,dy.
\end{equation}
The difference $T-t_s$ is often called {\it time-to-go} in control theory. Since $\lambda(t)=\mu(t)$ at the VR junction, we obtain a relation between the switching values $t_s,x_{s}$ and $y_{s}$ satisfied at such a junction. To express it, it is convenient to introduce new variables: 
\begin{gather}\label{NewVar}
\xi:=P(x)/P'(x),\ \ \ \tau:= P'(x)(T-t),\ \ \ y_T=y+\xi\tau, 
\end{gather}
and define the {\it switching function}:
\begin{align}\label{SwFun}
\Sw(\tau,\xi,y) &:= L(y_T)(y_T-y) - \int_{y}^{y_T} L(\zeta) d\zeta - \xi\big(L(y_T) - L(y)\big)\\
&=\int_{y}^{y_T} \zeta L'(\zeta) d\zeta-(\xi+y)\big(L(y_T) - L(y)\big)\,.\notag
\end{align}
Note that in $\tau$, $\xi$, $y$ coordinates the equation of the switching surface does not depend on $P(x)$. One can see that 
\begin{equation}\label{SwMarg}
\Sw(\tau_{s},\xi_s,y_s)=\frac1{\xi_sP(x_s)}\big(\lambda(y_s)-\mu(y_s)\big),
\end{equation}
i.e. the switching surface is given by the equation $\Sw=0$ with $\tau>0$. This equation allows one to find VR trajectories (almost) analytically. On the V arc we have $\dot{x}=P(x)$ and $\dot{y}=0$. This gives $y_s=y_0$ and a system of two equations to determine $t_s$ and $x_s$:
\begin{equation}\label{tsxs}
t_s=\int\limits_{x_0}^{x_s}\frac{dx}{P(x)};\ \ \  \Sw\left(P'(x_s)(T-t_s),\frac{P(x_s)}{P'(x_s)},y_s\right)=0.
\end{equation}
Once these values are found, the VR trajectory can be easily integrated, and satisfies the necessary conditions of the Pontryagin's theorem. Of course, these trajectories are not always optimal.

Let us prove some simple properties of the switching surface.
\begin{lemma}\label{SubR} A trajectory that passes through a point above the switching surface (i.e. $\Sw>0$), and is R afterwards, is suboptimal.
\end{lemma}
\begin{proof}
Without loss of generality, we may assume that at our point $t=0$. Let $J(x,y,T)$ denote the payoff computed along an R trajectory starting from $(x,y)$ over $[0,T]$. To show that such a trajectory is suboptimal we will produce an admissible trajectory with higher payoff. Consider a trajectory with $u = 1$ on $[0,\epsilon]$, and $u = 0$ on $[\epsilon,T]$. Denoting $x_\epsilon=x(t_\epsilon)$ we can express the new payoff as: 
\begin{equation*}
    J_\epsilon = L(y)\epsilon + J(x(\epsilon),y,T-\epsilon).
\end{equation*}
Now $\frac{dx(\epsilon)}{d\epsilon}\Big|_{\epsilon=0}=P(x)$,  $\frac{\partial J}{\partial T}=L(y_T)$, and $\frac{\partial J}{\partial x}=\lambda(0)$ computed along the R trajectory (by the standard properties of marginal values \cite[2.2]{LW}). Taking into account \eqref{AdjR}, we compute
$$
\frac{dJ_\epsilon}{d\epsilon}\Big|_{\epsilon=0} = P(x)\big(\lambda(0) - \mu(0)\big).
$$ 
Since $(x,y,0)$ is above the switching surface, \eqref{SwMarg} implies that $\lambda(0) > \mu(0)$, and the derivative is positive. But then $J_\epsilon>J_0$ for small $\epsilon>0$, as required.
\end{proof}
Next, we will show that for monotone non-increasing functions $L(y)$ the switching surface does not come too close to the plane $t=T$, i.e. the time-to-go at the VR junction cannot be too small. This is plausible intuitively since the plant needs enough time to convert accumulated $x$ into $y$. The inequality we will prove is sharp since, for $L(y) = y$ and $P(x) = x$, the switching surface is the plane $T-t = 2$ (VR switches in this case happen with a constant time-to-go as follows from \eqref{SwCurve} with $\alpha=1$, for example).
\begin{theorem}\label{tau2}
Suppose that $L'(y)$ is monotone non-increasing. Then the time-to-go on the switching surface (at the junction with the final R arc) satisfies $T-t_s\geq\frac2{P'(x_s)}$. This inequality is strict if $L'(y)$ is strictly decreasing.
\end{theorem}
\begin{proof} It will be convenient to use the ratio variable $r:=\frac{\xi}{y}$. We will first transform the second equation in \eqref{SwFun}. Writing $L(y_T) - L(y)=\int^{y_T}_y L'(\zeta)d\zeta$, and changing to a new variable $s=\frac{\zeta}y$ in both integrals, we obtain:
\begin{equation}\label{Swrs}
\Sw(\tau,\xi,y)=y^2L'(y)\left(\int_1^{1+r\tau}\!\!\!\!\!\!sL'(ys)ds-(1+r)\int_1^{1+r\tau}\!\!\!\!\!\!L'(ys)ds\right)\,.
\end{equation}
Next, we will need the following elementary inequality for positive monotone non-increasing functions $l(s)$:
$$
\int_1^R\!\!\!\!s\,l(s)ds\leq\frac{R+1}2\int_1^R\!\!\!\!l(s)ds
$$
Indeed, this is equivalent to $\displaystyle{\int_1^R\!\!\!\left(s-\frac{R+1}2\right)l(s)ds\leq0}$, which holds as the difference $s-\frac{R+1}2$ is symmetric with respect to the middle of the interval $[1,R]$, and larger values of $l(s)$ are on the side where it is negative. The equality is only attained when $l=\textrm{const}$.

Therefore, taking $l(s)=L'(ys)$ in \eqref{Swrs} we have that for the points on the switching surface:
$$
\int_1^{1+r\tau}\!\!\!\!\!\!sL'(ys)ds=(1+r)\int_1^{1+r\tau}\!\!\!\!\!\!L'(ys)ds
\leq\frac{2+r\tau}2\int_1^{1+r\tau}\!\!\!\!\!\!L'(ys)ds.
$$
Canceling the integrals and $r$, we obtain $2\leq\tau$. The conclusion now follows from \eqref{NewVar}.
\end{proof}

\subsection{Cutoff curve}\label{Cut}

Not every point on the switching surface can be the VR junction of an optimal trajectory. When it is, it follows from \eqref{Bang} that $\lambda>\mu$ reverses to $\lambda<\mu$ there. Therefore, at the junction points not only does $\lambda=\mu$, but also $\dot{\lambda}\leq\dot{\mu}$. 
\begin{definition} A point $(x,y,t)$ of the switching surface is said to be on the cutoff curve when the time derivatives of the costate variables computed along a pure R extremal trajectory starting at it are equal, $\dot{\lambda}=\dot{\mu}$. If $\dot{\lambda}<\dot{\mu}$ we say that the point is above the cutoff curve.
\end{definition}
\noindent Intuitively, the cutoff curve separates the ``active" part of the surface, where the optimal junctions can occur, from the ``inactive" part. Optimal VR junctions can only occur on or above the cutoff curve, hence the name. As we will see below, this curve serves as a ``cutoff" also for singular arcs. The cutoff condition can be expressed explicitly by using \eqref{AdjR}, and recalling that for $\lambda=\mu$ we have:
$$
\dot{\lambda}=-\mu P'(x_s)=-L'(y)=\dot{\mu}.
$$
Recall from \eqref{NewVar} that $y_T=y+\xi\tau$, and define the {\it cutoff function}:
\begin{equation}\label{CutFun}
\Cut(\tau,\xi,y) = L(y_T) - L(y) -\xi L'(y).
\end{equation}
Note that by the costate equations
\begin{equation}\label{CutMarg}
\Cut(\tau_{s},\xi_s,y_s) = -\frac1{\xi_s}\big(\dot{\lambda}(y_s)-\dot{\mu}(y_s)\big).
\end{equation}
Thus, points on the cutoff curve can be characterized as satisfying $\Cut=0$, and those above it as satisfying $\Cut>0$. 

As already mentioned, the cutoff curve is also a cutoff in another sense. As the only place where $\lambda=\mu$ and $\dot{\lambda}=\dot{\mu}$ on the final R arc, it is also the only place where an optimal junction with singular M arcs can occur. At these junctions, the singular control \eqref{SingControl} must be terminated and switched to $u=0$. Conversely, if the cutoff curve is empty, or the singular control on it is inadmissible, optimal trajectories can only be bang-bang. The next theorem gives sufficient conditions for the cutoff curve to be either empty or ``large".
\begin{theorem}\label{CutExist}
Suppose that $L(y)$ is strictly increasing and $L(y)\xrightarrow[y\to\infty]{}\infty$. If $L(y)$ is convex then the cutoff curve is empty. If $\frac{L(y)}y\xrightarrow[y\to\infty]{}0$ then for any $y>0$ there is a point on the cutoff curve with this value of $y$. If, moreover, $L(y)$ is strictly concave, this point is unique.
\end{theorem}
\begin{proof} Since $\Sw=\Cut=0$ on the cutoff curve, we get from \eqref{SwFun} and \eqref{CutFun}:
\begin{equation}\label{Sw+Cut}
    \int_y^{y_T}\zeta L'(\zeta)d\zeta = (\xi + y)\xi L'(y)
\end{equation}
By our assumptions, $L(y)$ has an inverse function $L^{-1}(z)$ defined for $z>0$, and $L^{-1}(z)\to\infty$ when $z\to\infty$. By making a change of variables in the integral, $\zeta=L(y)$, taking into account that $L(y_T) = L(y)+\xi L'(y)$, and denoting $\eta:=\xi L'(y)$, we get the following equation for $\eta$:
\begin{equation}\label{EtaCut}
\int^{L(y) + \eta}_{L(y)}\!\!L^{-1}(z) dz = \Big(\frac{\eta}{L'(y)} + y\Big)\eta.
\end{equation}
Let $F(\eta)$ and $G(\eta)$ denote the left and the right hand sides of \eqref{EtaCut}, respectively. Then we have $F(0)=G(0)=0$, and
\begin{align}\label{EtaDer}
&F'(\eta)\Big|_{\eta=0}=L^{-1}\big(L(y)+\eta\big)\Big|_{\eta=0}=y;\ \ \ \ \ \ \ \ \ \ \ \ \ \
G'(\eta)\Big|_{\eta=0}=y+2\frac{\eta}{L'(y)}\Big|_{\eta=0}=y\notag\\
&F''(\eta)\Big|_{\eta=0}=(L^{-1})'\big(L(y)+\eta\big)\Big|_{\eta=0}=\frac{1}{L'(y)};\ \ \ \ \
G''(\eta)\Big|_{\eta=0}=\frac{2}{L'(y)}.
\end{align}
Therefore, $F''(\eta)<G''(\eta)$, $F'(\eta)<G'(\eta)$, and $F(\eta)<G(\eta)$ for small $\eta>0$. When $L(y)$ is convex, $L^{-1}(z)$ is concave, and $(L^{-1})'(z)$ is non-increasing, so $(F-G)''(\eta)$ is as well. But then, it remains negative for all $\eta>0$, and so do $(F-G)'$ and $F-G$. Thus, there are no positive solutions to $F(\eta)=G(\eta)$.
\medskip 

When $\frac{L(y)}y\xrightarrow[y\to\infty]{}0$, we have $\frac{L^{-1}(z)}z\xrightarrow[y\to\infty]{}\infty$, and $\frac{F'(\eta)-G'(\eta)}\eta\xrightarrow[\eta\to\infty]{}\infty$. Therefore, $F(\eta)-G(\eta)\xrightarrow[\eta\to\infty]{}\infty$,  and $F(\eta)=G(\eta)$ for some $\eta>0$. Setting $$\xi:=\frac{\eta}{L'(y)},\ \  \tau:=\frac{L'(y)}{\eta}\left(L^{-1}\big(L(y)+\eta\big)-y\right)
 $$ 
 gives us the desired point on the cutoff curve.
\smallskip  

If $L(y)$ is strictly concave, then $L^{-1}(z)$ is strictly convex, and $(L^{-1})'(z)$ is strictly increasing. Then, so is $(F-G)''$. If it stayed negative for $\eta>0$ then so would $(F-G)'$ and $F-G$, contradicting the existence of a solution. Therefore, $F''(\eta_{**})=G''(\eta_{**})$ for some $\eta_{**}>0$. Similarly, $F'(\eta_{*})=G'(\eta_{*})$ for some $\eta_{*}\geq\eta_{**}$, and the difference is positive for $\eta>\eta_{*}$. Thus, $F-G$ is strictly increasing for $\eta>\eta_{*}$, and since it is negative for $\eta<\eta_{*}$, it must be zero at a single positive $\eta$.
\end{proof}
\begin{corollary}\label{ConGrowth}
If $L(y)$ is strictly increasing, convex, and $L(y)\xrightarrow[y\to\infty]{}\infty$, then the optimal control is always bang-bang, and the optimal trajectories are either R or VR.
\end{corollary}
\noindent This situation occurs in the case $L(y)=y$, and for $L(y)=y^\a$ with $\a\geq1$ generally. We see, however, that it is borderline to the range of gambler's fitness functions, while the more common concave fitness functions have non-empty cutoff curves. The increase  of $L(y)$ to $\infty$ is a technical condition, one can modify the proof to cover functions with a horizontal asymptote, like $L(y)=\frac{y^\a}{\a}$ for $\a<0$, see Lemma \ref{CutPt}.

\subsection{Singular surface}\label{Sing}

We are now in a position to define the final piece of our geometric configuration, the singular surface. When we substitute the singular control law \eqref{SingControl} into the state equations \eqref{StateEq} we obtain integral curves passing through every point of state space. However, these curves can only produce singular arcs of optimal trajectories if they pass through a point on the cutoff curve, where these arcs can be joined to the final R arc. This motivates the next definition. \begin{definition} We call the singular surface the set foliated by the integral curves of the system \eqref{StateEq} with the singular control \eqref{SingControl} (singular curves, for short) passing through points on the cutoff curve.
\end{definition}
\noindent By definition, the singular surface is empty when the cutoff curve is empty, and, when it is non-empty, it intersects the switching surface along the cutoff curve, see Figure \ref{SwSing}. The part of the singular surface past the cutoff curve is ``inactive", i.e. singular arcs of optimal trajectories can not pass through it there.

While the non-emptiness of the cutoff curve is a necessary condition for the existence of singular arcs, it is not sufficient. {\it A priori}, the singular curves may not even be locally optimal, or the singular controls on them may not be admissible, that is, remain within the bounds $0\leq u\leq1$. 
\begin{figure}[htbp]
\vspace{-0.15in}
\begin{centering}
\includegraphics[width=125mm]{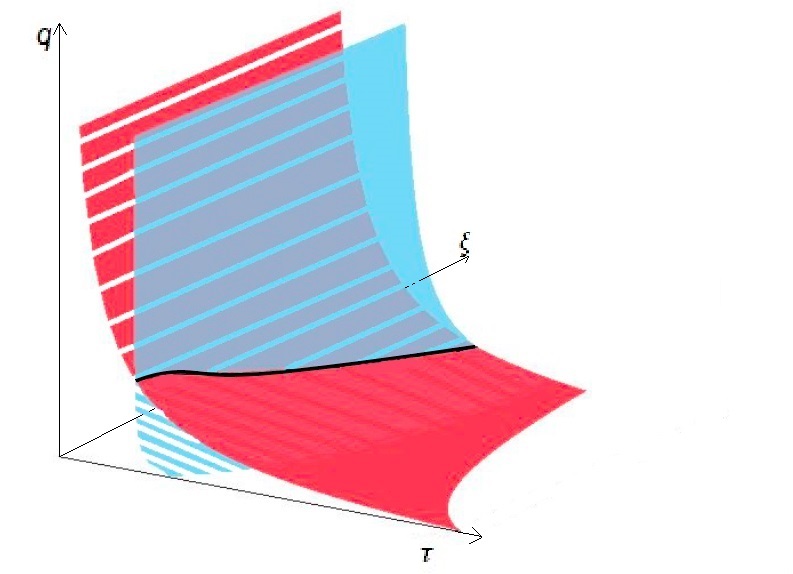}
\par\end{centering}
\vspace{-0.5in}
\hspace*{-0.1in}\caption{\label{SwSing} Schematic depiction of the switching surface (blue), the cutoff curve (black), and the singular surface (red) in $\xi$-$q$-$\tau$ coordinates, $\xi=\frac{P(x)}{P'(x)}$ and $q=\frac{y}{\xi}$. The lined parts are inactive.}
\end{figure}
We will now sketch how the geometric configuration, as depicted on Figure \ref{SwSing}, determines the nature of optimal schedules. To interpret this biologically, recall that our model only covers the volatile part of the season, with uncertain duration, so that the plant is likely to be mature already at the starting point. Note that the direction of $\tau$ is opposite to the direction of time $t$, so the progression of time corresponds to moving to the left in the figure.

Trajectories that start under the singular surface, or to the left of the switching surface, move up along an R arc until they meet the active part of the singular surface, or reach $\tau=0$ ($t=T$). In the latter case, they are pure R, in the former, they switch to a singular M arc, and follow it until meeting the cutoff curve. When there, they switch to a final R arc, making them RMR. Trajectories that start above both surfaces move down along a V arc until meeting one of them (or both, on the cutoff curve). At the switching surface they switch to a final R arc (VR), and at the singular surface they first switch to a singular M arc, and then to a final R arc at the cutoff curve (VMR). MR trajectories must start on the active part of the singular surface. 
Of course, this intuitive picture depends on the configuration being roughly as depicted. In the next section, we will confirm that this is indeed the case for the power fitness functions, and add some analytic characterizations.

\section{Power fitness functions}\label{Power}

In this section we apply geometric considerations of the previous section to  the optimal allocation problem \eqref{StateEq}-\eqref{Pay} with the power fitness functions $L(y) = \frac{y^\alpha}{\alpha}$. The $\frac1{\alpha}$ is added for technical convenience and does not affect optimization. It simplifies some formulas, provides the correct sign for $\alpha<0$, and ensures continuity with $L(y) = \ln y$ for $\alpha\to0$, as $\frac{y^\alpha-1}{\alpha}\xrightarrow[\alpha\to0]{}\ln y$.

The geometric configuration we introduced simplifies considerably when $L(y)$ is a power function or logarithm. The underlying reason is that those are exactly the functions for which the fitness density $L'(y)$ is multiplicative, i.e. $L'(ys)=L'(y)L'(s)$, and the scaling factor $l_y(s):=\frac{L'(ys)}{L'(y)}$ does not depend on $y$. As a result, a prominent role is played by the ratio $r:=\frac{\xi}{y}$ and its reciprocal $q:=\frac{y}{\xi}$. For power production functions $P(x)=kx^\beta$ we have $\xi=\frac{x}{\beta}$, so these ratios are proportional to the ratios of vegetative and reproductive masses. We will refer to $q$ as the {\it mass ratio} in general. Much of the geometry can be described in two-dimensional terms, using $q$,$\tau$ or $r$,$\tau$ as coordinates.

\subsection{Switching curve}

Let us rewrite the switching equation $\Sw=0$ in terms of $l_y(s)$:
\begin{equation}\label{SwrsPow}
\int_1^{1+r\tau}\!\!\!\!\!\!sl_y(s)ds=(1+r)\int_1^{1+r\tau}\!\!\!\!\!\!l_y(s)ds;
\end{equation}
Since $l_y(s)=l(s)$, the switching surface is a cylinder over a curve in the $r$-$\tau$ plane given by \eqref{SwrsPow}. For $\alpha\neq0,-1$ it simplifies to\footnote{For $\alpha=0$ this is replaced by $\tau=\left(1+\frac1r\right)\ln(1+r\tau)$, and for $\alpha=-1$ by $\tau=\frac1{r(1+r)}(1+r\tau)\ln(1+r\tau)$.}
\begin{equation}\label{SwCurve}
\tau=\left(1+\alpha+\frac1r\right)\frac{1-(1+r\tau)^{-\alpha}}\alpha.
\end{equation}

We will call the curve given by this equation the {\it switching curve}. For $\alpha=1$ it reduces to $\tau=2$, i.e. the switching curve is a vertical line. To determine its shape in general, it is convenient to rewrite \eqref{SwCurve} in a parametric form using $z=r\tau$ as the parameter (we leave the $\alpha=0,-1$ cases to the reader):
\begin{align}\label{ParSwCurve}
r(z)&=\frac1{1+\alpha}\frac{(1+z)^{-\alpha}-1-\alpha z}
{(1+z)^{-\alpha}-1};\\
\tau(z)&=(1+\alpha)\,\frac{z\left((1+z)^{-\alpha}-1\right)}{(1+z)^{-\alpha}-1-\alpha z}.
\end{align}
One can prove the following by direct calculations. 
\begin{lemma}\label{SwCurv}
The switching curve is the graph of a positive, monotone decreasing function $q(\tau)=\frac1{r(\tau)}$ on $(2,\infty)$. It has a vertical asymptote at $\tau=2$, and a horizontal asymptote when $\tau\to\infty$ (and $z\to\infty$). The latter is $q=0$ for $-1<\alpha<1$ (and the special cases $\alpha=0,-1$), and $q=-\alpha-1$ for $\alpha<-1$. 
\end{lemma}
\noindent For $\alpha=-2$ one can derive an explicit equation, $q(\tau)=\frac{\tau}{\tau-2}$, see Figure \ref{Alpha-2}.

\subsection{Cutoff point}\label{CutPow}

In terms of $l_y(s)$, the cutoff equation $\Cut=0$ becomes
\begin{equation}\label{Cutrs}
r=\int_1^{1+r\tau}\!\!\!\!\!\!l_y(s)ds.
\end{equation}
Accordingly, the cutoff curve is (generically) a collection of straight lines over points solving the system \eqref{SwrsPow},\eqref{Cutrs}. We will show when there is, in fact, such a point which we will call the {\it cutoff point}. 
\begin{lemma}\label{CutPt}
For $\alpha\geq1$ there are no cutoff points. For $\alpha<1$ there is a unique point $(r_c,\tau_c)$, with $\tau_c\geq2$, $r_c\geq\frac1{1-\alpha}$. For $\alpha<0$, the unique point also satisfies $r_c<-\frac1{\alpha}$.
\end{lemma}
\begin{proof}
For $\alpha\neq0,-1$ the system \eqref{SwrsPow},\eqref{Cutrs} simplifies to\footnote{For $\alpha=0$ the second equation is replaced by $r=\ln(1+r\tau)$, and for $\alpha=-1$ by $r(1+r)=\ln\tau$.}
\begin{align}\label{Powrt}
\tau&=1+\frac{r}{1+\alpha r};\notag\\
1+\alpha r&=(1+r\tau)^\alpha.
\end{align}
It follows from the second equation that $1+\alpha r\geq0$ for any positive solution, and the system reduces to a single equation for $r$:
\begin{equation}\label{Cutr}
f_\alpha(r):=(1+\alpha r)^{\frac1\alpha}-\left(1+r+\frac{r^2}{1+\alpha r}\right)=0.
\end{equation}
\vspace{-0.2in}
\begin{figure}[!ht]
\begin{centering}
\includegraphics[width=8cm]{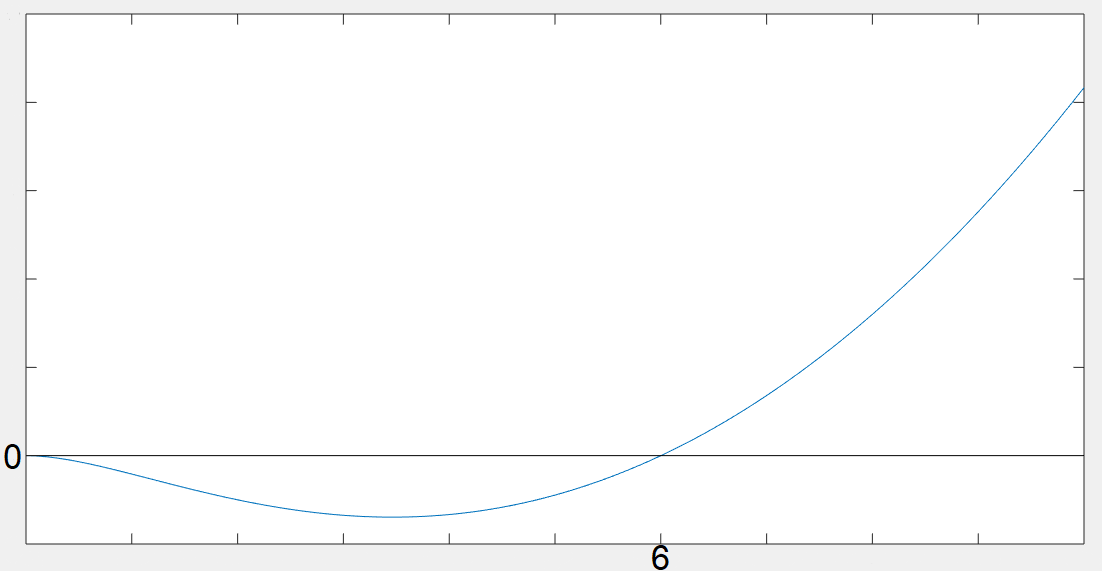}
\par\end{centering}
\vspace{-0.2in}
\hspace*{-0.1in}\caption{\label{falpha} The graph of  $f_\alpha(r)$ for $\a=\frac{1}{2}$, the shape of which is typical for $0\leq\alpha<1$. For $-1<\alpha<0$ the graph has a vertical asymptote at $r=-\frac1\alpha$, and for $\alpha<-1$ its shape is inverted about the $r$-axis.}
\end{figure}
\medskip 

The rest of the proof is similar to the proof of Theorem \ref{CutExist}, and we only sketch it. Note that $f_\alpha(0)=f'_\alpha(0)=0$, and $f''_\alpha(0)<0$, for all $\alpha$. Moreover (we leave the case $\alpha=0$ to the reader),
$$
f_\alpha^{''}(r)=(1-\alpha )(1+\alpha r)^{\frac{1-2\alpha}{\alpha}}-\frac{2}{(1+\alpha r)^{3}}
$$
stays negative for all $r>0$ when $\alpha\geq1$. Therefore, there are no positive solutions in this case. For $-1<\alpha<1$ the equation $f_\alpha^{''}(r)=0$ has a single positive zero $r_{**}$ that can be found explicitly. For $0<\alpha<1$ this means that $f_\alpha^{'}(r)$ decreases on $(0,r_{**})$ and strictly increases on $(r_{**},\infty)$. Since it is negative on $(0,r_{**})$ there is a unique $r_*$ on $(r_{**},\infty)$ such that $f_\alpha^{'}(r_*)=0$. 

By a similar argument, $f_\alpha(r)=0$ has a unique positive root $r_c\geq r_*\geq r_{**}>0$. The estimate $\tau_c\geq2$ follows from Theorem \ref{tau2}, and implies $r_c(1-\alpha)\geq1$ by the first equation in \eqref{Powrt}. For $-1<\alpha<0$ one should replace $\infty$ with $-\frac1{\alpha}$ and for $\alpha<-1$ reverse the positive/negative and the increase/decrease in the argument above.
\end{proof}

\noindent The case $\alpha=0$, i.e. $L(y)=\ln y$, was studied by King and Roughgarden, who estimated $r_c\approx1.793$ and $\tau_c\approx2.793$. When $\alpha\to1$ our estimate implies that $r_c\to\infty$, and $\tau_c\to2$, and one can show that $r_c\sim-\frac1{\alpha}$ for large negative $\alpha$. Graphs of $r_c(\alpha)$ and $\tau_c(\alpha)$ are depicted on Figure \ref{rtau}. Biologically, this means that, all else being equal, increase in the fitness index $\a$ leads to shorter times-to-go at the switch to the final reproductive arc, and larger  vegetative mass at the time of the switch.
\begin{figure}[!ht]
\vspace{-0.1in}
\begin{centering}
\includegraphics[width=90mm]{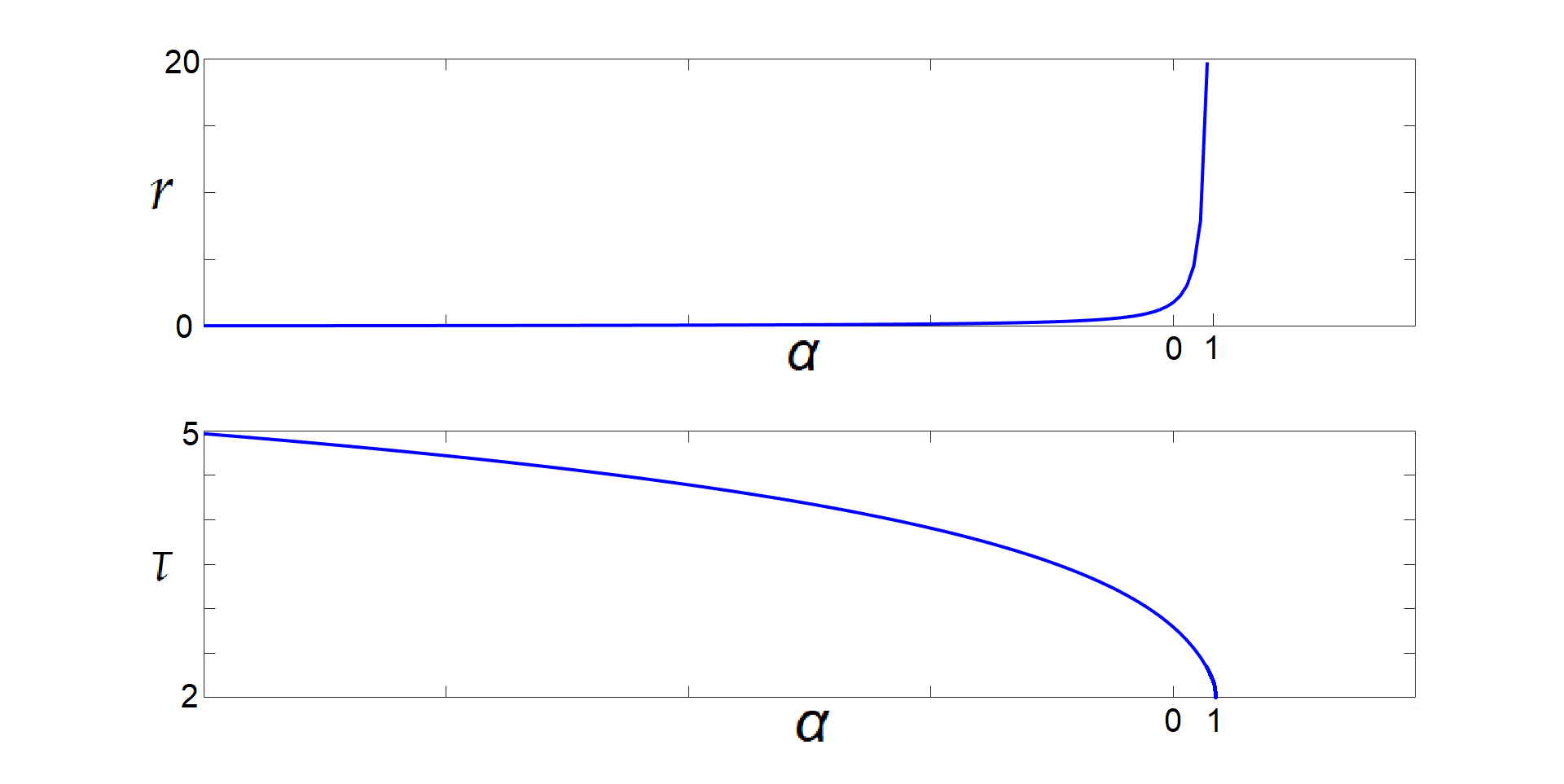}
\par\end{centering}
\vspace{-0.3in}
\hspace*{-0.1in}\caption{\label{rtau} Graphs of $r_c(\alpha)$ (above) and $\tau_c(\alpha)$ (below) on $(-\infty,1)$.}
\end{figure}

Recall that points on the switching surface that are above the cutoff curve are the points where optimal control can switch from $u=1$ to $u=0$. In general, these are characterized by the inequality $\Cut>0$, where $\Cut$ is defined in \eqref{CutFun}. We can now make this condition more explicit.
\begin{lemma}\label{VRabove} If $(\tau,r,y)$ is a VR junction point of an optimal trajectory then $r\leq r_c$ and $\tau\leq\tau_c$.
\end{lemma}
\begin{proof}
The $\Cut\geq0$ condition simplifies to 
\begin{equation*}\label{derivs}
    r\leq\frac{(1+r\tau)^\alpha - 1}{\alpha}.
\end{equation*}
Combining this inequality with \eqref{SwCurve}, we derive:
\begin{equation*}
    r\tau-\frac{(1+r\tau)^{\alpha}-1}{\alpha}\left(1+\frac{1-(1+r\tau)^{-\alpha}}{\alpha}\right)\leq0.
\end{equation*}
The same analysis as for $f_\alpha(r)$ in Lemma \ref{CutPt} can be carried out for the left hand side of this inequality as a function of $r\tau$. It has a single positive zero $r_c\tau_c$, and the inequality implies that $r\tau\leq r_c\tau_c$. By \eqref{Powrt}, we obtain: 
$$
r + \frac{r^2}{1+\alpha r}\leq r_c + \frac{r_c^2}{1 + \alpha r_c}.
$$
Since the left hand side is strictly increasing for positive $r$ (for $r<-\frac1\a$ when $\a<0$) this implies $r\leq r_c$. The second estimate follows from monotonicity of $\tau$ as a function of $r$ given implicitly by \eqref{SwCurve}.
\end{proof}

\subsection{Existence of singular arcs}\label{ExiSing}

Note that $r_c, \tau_c$ are universal functions of $\alpha$ that do not depend on the initial values, the season length $T$, or the production function $P(x)$. This means that the switching curve and the cutoff point do not depend on $P(x)$ either. However, the two-dimensionalization of the geometry described above is incomplete for general production functions. The singular control, and therefore the singular surface, does depend on $P(x)$, and is not a cylinder over a curve in the $r$-$\tau$ plane in general. Nonetheless, the information on the active part of the switching surface, provided by the last lemma, along with monotone behavior of the mass ratio $q$ along V and R arcs, is sufficient to infer the existence of singular arcs on optimal trajectories.

The inconvenience of the smaller $r$ values being ``above" the cutoff curve when is removed by using $q=\frac1r$, which makes ``above" indeed above, see Figure \ref{PowSwSing}. It follows from Lemma \ref{VRabove} that the region $q>q_c$, $\tau>\tau_c$, is above the switching surface, which leads to the main result of this section.
\begin{figure}[htbp]
\vspace{-0.15in}
\begin{centering}
\includegraphics[width=125mm]{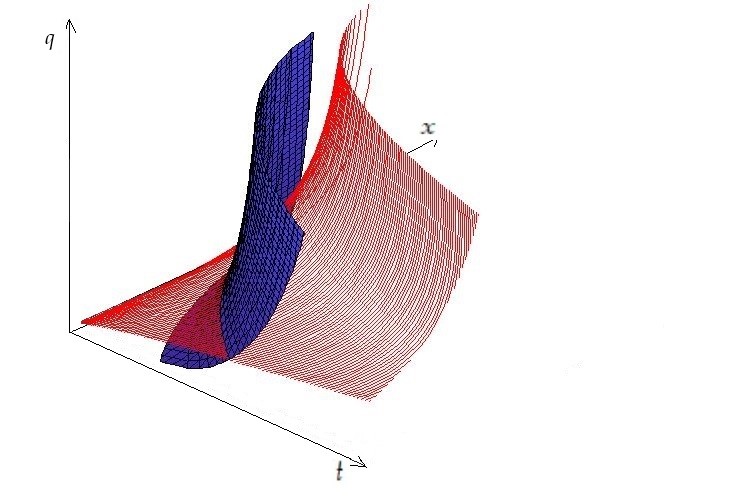}
\par\end{centering}
\vspace{-0.5in}
\hspace*{-0.1in}\caption{\label{PowSwSing} The switching surface (blue), and the singular surface (red) for $P(x)=x^{0.7}$ and $L(y)=2y^{0.5}$.}
\end{figure}

\begin{theorem}\label{Mexists} Let $L(y) = \frac{y^\alpha}{\alpha}$, $\alpha<1$, and $P(x) $ be concave down. Set $t_*:=T-\frac{\tau_c}{P'(x_0)}$, and suppose that  
\begin{equation}\label{Mbounds}
q_c-P'(x_0)t_*<q_0<q_c\,e^{P'(x_0)t_*}\,.
\end{equation}
Then, any optimal trajectory with $q_0$ as the initial value has a singular arc.
\end{theorem}
\begin{proof}If the initial arc is M, we are done. If it is R then $u=0$, and $\dot{q}=\frac{P(x_0)}{\xi}=P'(x_0)$. By the lower bound in \eqref{Mbounds}, we have that $q(t_*)>q_c$ and $\tau(t_*)=P'(x_0)(T-t_*)=\tau_c$. This means that just before $t_*$ the arc is above the switching curve, and hence is suboptimal past $t_*$ by Lemma \ref{SubR}. Since RV switches are impossible by Theorem \ref{5Structures}, it must switch to an M arc before $t_*$.
\smallskip 

If the initial arc is V, we can estimate:
\begin{equation*}
\dot{\xi}=\left(\frac{P(x)}{P'(x)}\right)'\dot{x}=\left(1-\frac{P''(x)}{P'(x)}\frac{P(x)}{P'(x)}\right)P(x)=\left(P'(x)-P''(x)\xi\right)\xi\geq P'(x)\xi.
\end{equation*}
Therefore, $\xi(t)\geq\xi_0e^{P'(x)t}$ and $q(t)=\frac{y_0}{\xi(t)}\leq q_0e^{-P'(x_0)t}$. This means that $\tau(t)<\tau_c$ for $t<t_*$, and, by the upper bound in \eqref{Mbounds}, $q(t)<q_c$ for $t\geq t_*$. Therefore, by Lemma \ref{VRabove}, this V arc can not intersect the switching surface at a point where the VR switch is optimal. Since the final arc must be R, it must switch to an M arc before $t_*$. 
\end{proof}
\begin{corollary}\label{Marcs} Under the conditions of Theorem \ref{Mexists}, if $T>\frac{\tau_c}{P'(x_0)}$ there exist optimal trajectories that include singular arcs, and if $T>\frac{\tau_c+q_c}{P'(x_0)}$ then pure R trajectories are suboptimal, and all optimal trajectories with small enough $q_0$ are RMR.
\end{corollary}
In other words, singular arcs (mixed growth) are a general phenomenon for conservative ($\alpha<1$) fitness indices. On the other hand,  for $\alpha\geq1$, that corresponds to ``gambler's" strategies, we always get bang-bang behavior. 

\section{Linear production and the mass ratio diagram}\label{Lin}

In this section we give a complete solution of the optimal allocation problem for linear production functions and power utilities. Note that for linear $P(x)$ we have $\xi=x$ and $q=\frac{y}{x}$, which is merely the ratio of the reproductive to the vegetative mass of the plant. The switching and the singular surfaces are both cylindrical in $q$-$t$ coordinates, so the solutions can be depicted on a a two-dimensional diagram, which we call the {\it mass ratio diagram}, see Figure \ref{KRdiagram}. For $L(y)=\ln y$, this diagram was originally described (with only partial proof) in \cite{KR}. It is somewhat remarkable that essentially the same diagram is valid for all $L(y) = \frac{y^\alpha}{\alpha}$ with $\alpha<1$, and even that the picture is two-dimensional for $\alpha\neq0$, as there is no obvious way to produce a payoff functional equivalent to \eqref{Pay} that depends only on the mass ratio. For $\alpha=0$ it can be done by using the additive property of logarithms.

For definiteness, we take $P(x)=x$, but one can easily reformulate the results for $P(x)=kx$ with $k>0$ by rescaling time. The equation of the switching curve in $q$-$t$ coordinates follows directly from \eqref{SwCurve}:
\begin{equation}\label{LinSwCurve}
T-t=(1+\alpha +q)\frac{1-(1+\frac{T-t}{q})^{\alpha}}{\alpha}\,.
\end{equation}
The state equations \eqref{StateEq} can be reduced to a single equation for the mass ratio $q$:
\begin{equation}\label{Statev}
\dot{q} = 1-(1+q)u.
\end{equation}
The singular control \eqref{SingControl} also depends only on the mass ratio:
\begin{equation}
u = 1-\frac1{1-\alpha}\frac{y}{x}=1-\frac{q}{1-\alpha}.
\end{equation}
\begin{figure}[!ht]
\vspace{-0.1in}
\begin{centering}
\ \ \ \ \ \ \ \includegraphics[width=80mm]{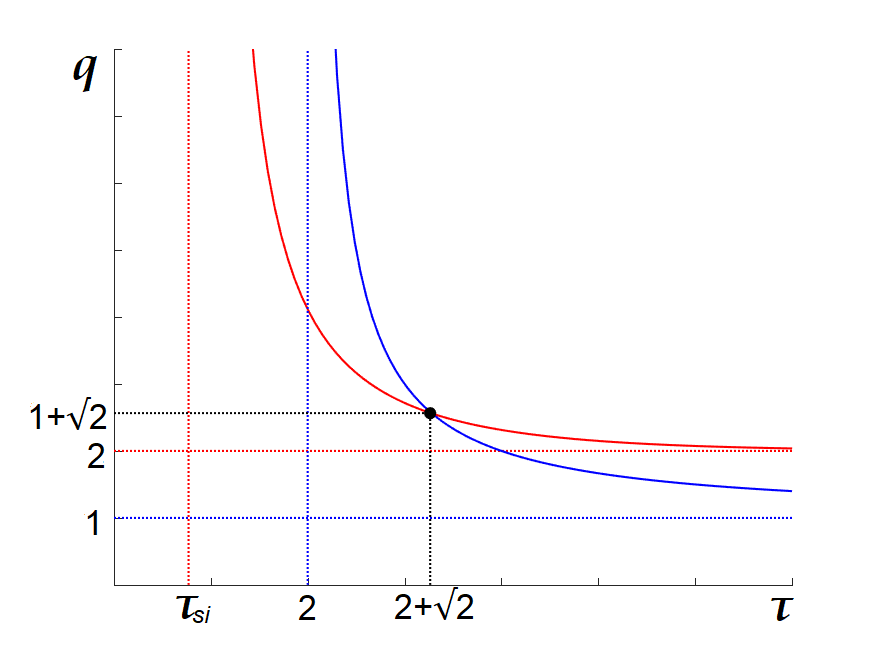}
\par\end{centering}
\vspace{-0.5in}
\hspace*{-0.1in}\caption{\label{Alpha-2} Switching (blue) and singular (red) curves for $\a=-2$ in $q$-$\tau$ coordinates, the shapes of which are typical for $\alpha<-1$. For $-1\leq\alpha<1$ the horizontal switching asymptote is $q=0$, and for $0\leq\alpha<1$ the horizontal singular asymptote is also $q=0$. The switching curve is universal for all production functions, the singular curve is for $P(x)=x$.}
\end{figure}
Integrating \eqref{Statev} with the singular control, we find for $\alpha\neq0$ that\footnote{For $\alpha=0$ this simplifies to $q(t)=\frac1{C-t}$.} 
\begin{equation}\label{SingCurve}
q(t)=\frac{\alpha}{e^{\frac{\alpha}{1-\alpha}(C-t)}-1}\,.
\end{equation}
The constant of integration $C$ can be found from the condition that the singular curve passes through the cutoff point, which gives $C=T-\tau_{si}(\alpha)$, where
\begin{align}\label{taus}
\tau_{si}(\alpha)&=\tau_c(\alpha)-(1-\alpha)\ln(1+r_c(\alpha)\tau_c(\alpha))\\
&=1+\frac{r_c(\alpha)}{1+\alpha\,r_c(\alpha)}
-\frac{1-\alpha}{\alpha} \ln\Big(1+\alpha\,r_c(\alpha)\Big),
\end{align}
and we explicitly indicated the dependence of $r_c,\tau_c$ from Lemma \ref{CutPt} on $\alpha$. Clearly, $\tau_{si}(\alpha)\leq\tau_c(\alpha)$ for $\alpha\leq1$, and $\tau_{si}(0)=1$. One can show that $\tau_{si}(\alpha)\geq1$ for $0\leq\alpha<1$. For $\alpha=-2$ we find that $\tau_{si}=2+\sqrt{2}-3\ln(\sqrt{2}+1)\approx0.77$, see Figure \ref{Alpha-2}.
\begin{figure}[htbp]
\begin{centering}
\includegraphics[width=140mm]{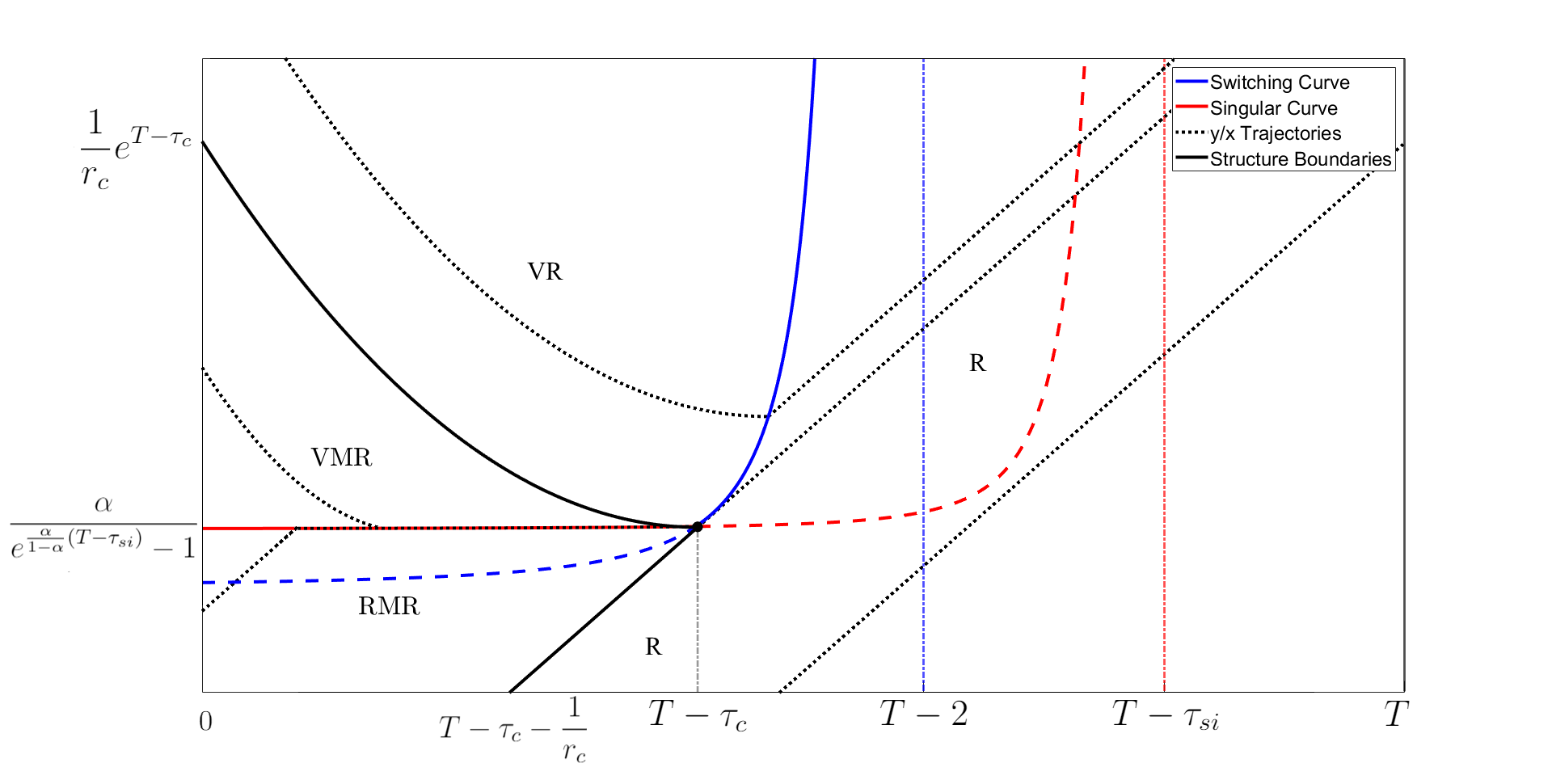}
\par\end{centering}
\vspace{-0.4in}
\hspace*{-0.1in}\caption{\label{KRdiagram} The mass ratio diagram for $L(y) = \frac{y^\alpha}{\alpha}$ with $\alpha<1$, and $P(x)=x$. The inactive parts of the switching and singular curves are dashed, and the vertical lines are their asymptotes. V arcs fall exponentially, R arcs rise with slope $1$, and M arcs follow the red curve. Optimal trajectories starting in the labeled regions have the indicated control structure, and MR trajectories must start on the solid part of the red curve.}
\end{figure}

The resulting configuration of the switching and singular curves, along with the arrangement of optimal trajectories it implies, is depicted on Figure \ref{KRdiagram}. We will now show how to determine the optimal schedules explicitly from the initial values based on it. The proof is a more precise version of the sketch given for the general case at the end of Section \ref{Sing}. Optimal solutions are found analytically by integrating the state equations. For simplicity, we state the theorem for the case of $T>q_c+\tau_c$, but the modifications needed for simpler diagrams, where pure R trajectories occur, are straightforward. We also leave out the modifications necessary for $\alpha=0,-1$.
\begin{theorem}\label{4LinStructures} Let $P(x) = x, L(y) = \frac{y^\alpha}{\alpha}$, $\alpha<1$, $q_0=\frac{y_0}{x_0}$ be the mass ratio of the initial values, and suppose that $T>q_c+\tau_c$. Then:
\smallskip

\textup{(i)} If $q_0<\frac{\alpha}{\exp(\frac{\alpha}{1-\alpha}\left(T-\tau_{si})\right)-1}$, then the optimal trajectory is RMR. The second switching time is $t_2=T-\tau_c$, and the first is $t_1=\frac{y_1-y_0}{x_0}$, where $y_1$ can be found by solving for $r_1:=\frac{x_0}{y_1}$ in
\begin{equation}\label{FirstMass}
\ln\frac{1+\alpha r_1}{1+\alpha r_0}=\frac{\alpha}{1-\alpha}\left(T-\tau_c-\ln\frac{r_1}{r_0}\right);
\end{equation}
\smallskip

\textup{(ii)} If $q_0=\frac{\alpha}{\exp(\frac{\alpha}{1-\alpha}\left(T-\tau_{si})\right)-1}$, then the optimal trajectory is MR. The switching time is $T-\tau_c$;
\smallskip

\textup{(iii)} If $\frac{\alpha}{\exp(\frac{\alpha}{1-\alpha}\left(T-\tau_{si})\right)-1}<q_0<\frac{1}{r_c}e^{T-\tau_c}$, then the optimal trajectory is VMR. The second switching time is $t_2=T-\tau_c$, and the first is $t_1=\int\limits_{x_0}^{x_1}\frac{dx}{P(x)}$, where $x_1$ can be found by solving for $r_1:=\frac{x_1}{y_0}$ in \eqref{FirstMass};
\medskip

\textup{(iv)} If $q_0\geq\frac{1}{r_c}e^{T-\tau_c}$, then the optimal trajectory is VR. The switching time $t_s$ can be found as $t_s=\ln\frac{x_s}{x_0}$ after solving for $x_s$ in:
\begin{equation}\label{Linxs}
T-\ln\frac{x_s}{x_0}=\left(1+\alpha+\frac{x_s}{y_0}\right)\frac1\alpha\left(1-\left(1+\frac{x_s}{y_0}\ln\frac{x_s}{x_0}\right)^{-\alpha}\right).
\end{equation}
\end{theorem}
\begin{proof}
One easily checks that the conditions of Theorem \ref{5Structures} and Lemma \ref{SubR} are met in this case.
\smallskip

\textup{(i)} The inequality means that the trajectory starts under the singular curve. The initial arc can not be V because V arcs drop exponentially, and no switch to any other arc is possible then, while having a pure V trajectory contradicts Theorem \ref{5Structures}(i). Since $q_0$ is not on the singular curve the initial arc must be R. It can either stay R to the end or switch to an M arc only on the cutoff curve. The former possibility is suboptimal by Lemma \ref{SubR}, so the second arc is M, which must switch to the final R arc by Theorem \ref{5Structures}. Equation \eqref{FirstMass} follows from \eqref{SingCurve} rewritten for $r=\frac1q$.
\smallskip

\textup{(ii)} Initial V and R arcs are ruled out by the reasoning in (i) and (iii), respectively. Hence, MR is the only possibility.
\smallskip

\textup{(iii)} The trajectory starts above the singular curve, and hence above the switching curve. If the initial arc were R it would rise, and hence it could not switch to M. Therefore it must be R all the way, as RV switches are ruled out by Theorem \ref{5Structures}. But this is also ruled out by Lemma \ref{SubR}. Therefore, the initial arc is V. It follows from the second inequality that it never crosses the active part of the switching curve. Since it can not remain V all the way it must switch to M when crossing the singular curve, and then to the final R arc. 
\smallskip

\textup{(iv)} Initial R arc is ruled out as in (iii), and the initial arc is V. The inequality implies that it never crosses the active part of the singular surface (except at the cutoff point in the case of equality). Therefore, it can only be VR. The equations specialize \eqref{tsxs} and  \eqref{SwCurve} to this case.
\end{proof}

\section{Full lifetime}\label{Ext}

The model that we studied so far had the length of the safe period $T_0$ set to $0$. In general, the Hamiltonian and the costate equations retain the same form on $[T_0,T_0+T]$, but on $[0,T_0]$ the Hamiltonian reduces to 
\begin{equation}\label{ExtHam}
H(x,y,\lambda,\mu,u)=\big(\mu + (\lambda - \mu)u\big)P(x),
\end{equation}
and the costate equation for $\mu$ trivializes:
\begin{equation}\label{ExtAdj}
\begin{cases}
\dot{\lambda} = -\big(\mu + (\lambda - \mu)u\big)P'(x)\\ 
\dot{\mu} = 0.\\
\end{cases}
\end{equation}
This simplifies the structure of optimal trajectories on $[0,T_0]$. We continue to use the labeling convention of Definition \ref{ArcLet}. The next theorem complements Theorem \ref{5Structures} for the full lifetime model.
\begin{theorem}\label{5StructuresExt}
Let $P(x)$ and $L(y)$ be twice differentiable and $P(x)$ be positive for $x,y>0$. Then, we have for the optimal trajectories of of problem \eqref{StateEq}, \eqref{ExtPay}:
\smallskip 

{\textup{(i)}} Optimal trajectories have no M arcs on $[0, T_{0}]$.
\smallskip 

{\textup{(ii)}} If an optimal trajectory arrives at $T_0$ along a V arc then this arc is initial.
\smallskip 

{\textup{(iii)}} If an optimal trajectory arrives at $T_0$ along an R arc then either this arc is initial, or there is a single switch on $[0, T_{0}]$ from the initial V arc.
\smallskip 

{\textup{(iv)}} If $P(x)$ and $L(y)$ meet all the conditions of Theorem \ref{5Structures} \textup{(iv)} then any optimal trajectory has no more than three switches, and one of the following switching structures: R, VR, VMR, RMR, VRMR.
\end{theorem}
\begin{proof}
{\textup{(i)}} Along a singular arc $\dot{\lambda} = -\mu P'(x)=\dot{\mu} = 0$. Since $P'(x)>0$ this implies $\lambda=\mu=0$. But $\lambda$ and $\mu$ are non-increasing over the entire interval $[0,T_0+T]$, so the singular arc would have to extend all the way to $T_0+T$, contradicting Theorem \ref{5Structures} (i) (the final arc is R).
\begin{figure}[htbp]
\vspace{-0.1in}
\begin{centering}
(a)\includegraphics[width=68mm]{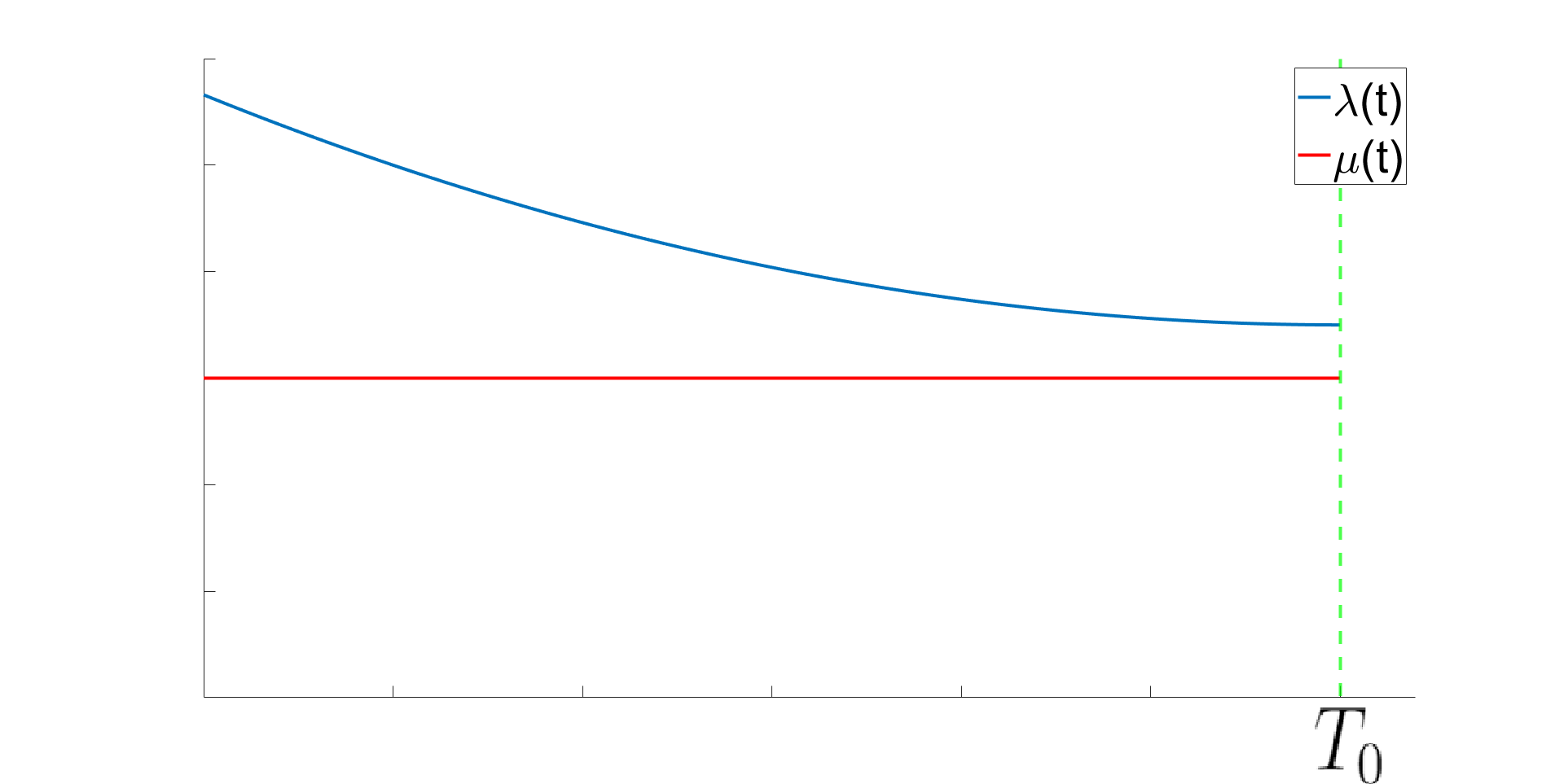} \hspace{0.0in} (b) \includegraphics[width=68mm]{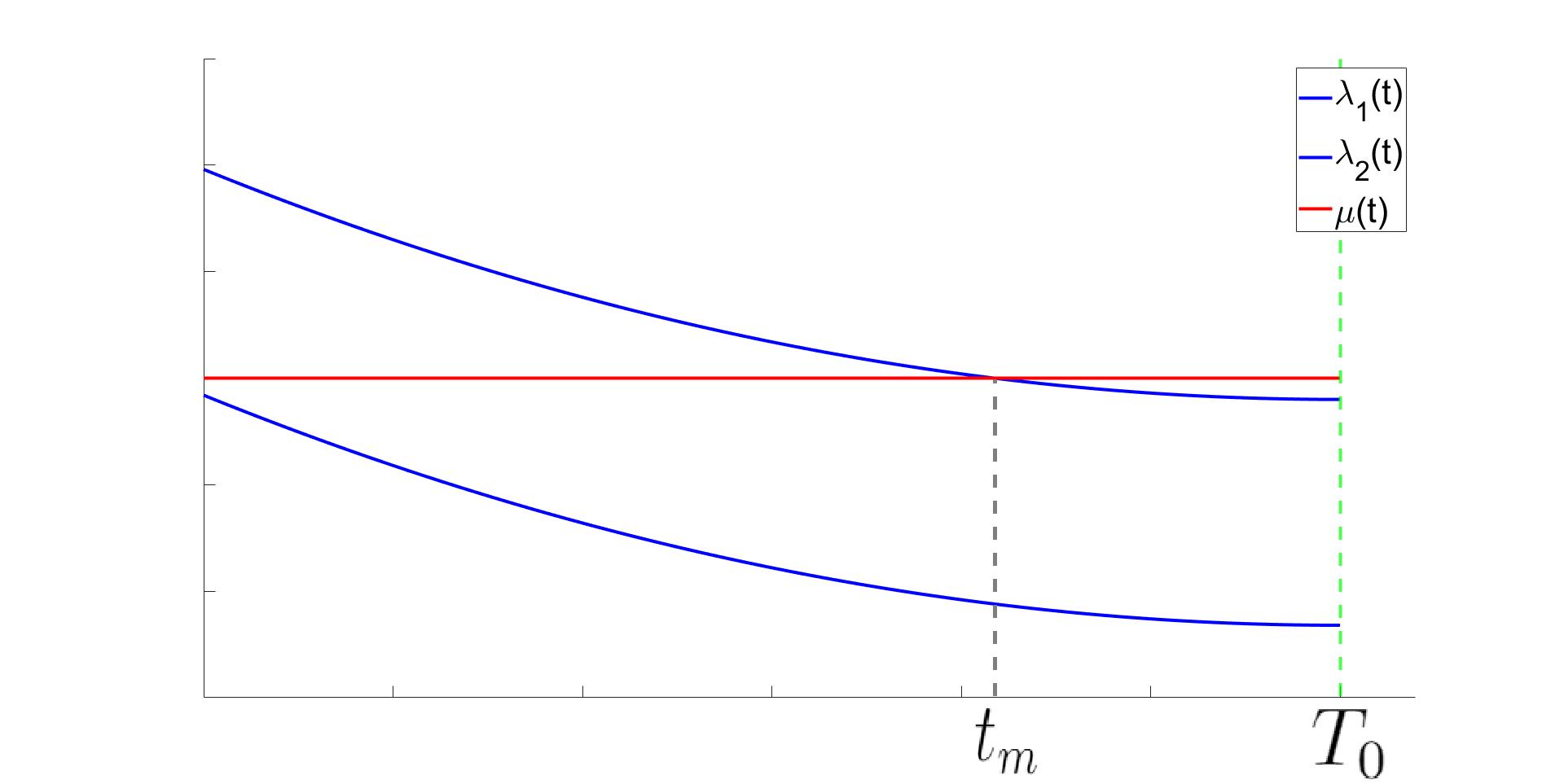}
\par\end{centering}
\vspace{-0.2in}
\hspace*{-0.1in}\caption{\label{Safelm} Marginal values on $[0, T_{0}]$ when (a) $\lambda(T_{0}) > \mu(T_{0})$ (V arc); (b) $\lambda(T_{0})<\mu(T_{0})$ for two different cases: VR (upper curve) and R (lower curve).}
\end{figure}
\smallskip 

(ii) and (iii). By the above argument, $\mu(T_{0})>0$, so $\mu$ is a positive constant on $[0, T_{0}]$. Since $(\lambda - \mu)u\geq0$ for any admissible control $u$,  $\lambda$ is strictly monotone decreasing there by \eqref{ExtAdj}. If the arc of arrival is V, then $\lambda(T_{0})\geq\mu(T_{0})$ and $\lambda(t)>\mu(t)$ for all $t<T_0$. Hence, there is no switch, see Figure \ref{Safelm} (a). If the arc of arrival is R then $\lambda(t)=\mu(t)$ at at most one point, see Figure \ref{Safelm} (b).
\smallskip 

{\textup{(iv)}} If the first arc on $[T_0,T_0+T]$ is M or V then we have a V arc on $[0, T_{0}]$ by (ii). Then, MR, VR, and VMR trajectories from Theorem \ref{5Structures} turn into VMR, VR, and VMR, respectively. If the first arc is R then by (iii) they either remain R on $[0, T_{0}]$, or are complemented by an initial V arc. This means that R trajectories from Theorem \ref{5Structures} turn into R or VR, while RMR turn into RMR or VRMR.
\end{proof}

The main effect of adding the safe period is, therefore, that initial V arcs are added, or already occurring ones extended, and there is, potentially, an additional switching time $t_m$ on $(0, T_{0}]$. The points $(t_m,x_m,y)$ form an additional switching surface that connects to the singular surface at $t=T_0$, see Figure \ref{PowSwSingExt}. If one can solve the problem on $[T_0,T_0+T]$ it is now easy to extend the solution to the entire interval. 
\begin{figure}[htbp]
\vspace{-0.1in}
\begin{centering}
(a)\includegraphics[width=68mm]{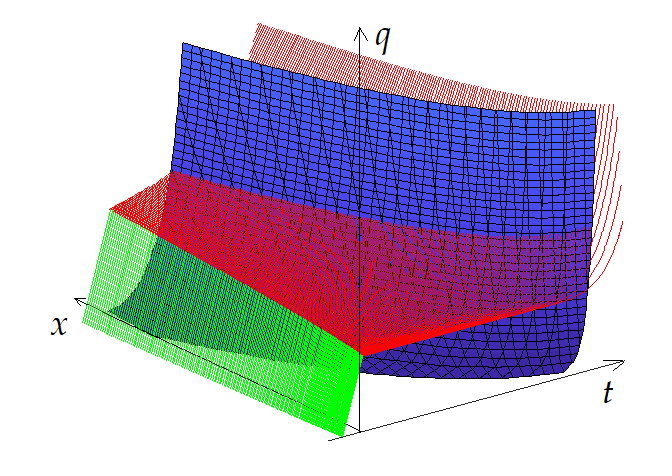} \hspace{0.0in} (b) \includegraphics[width=68mm]{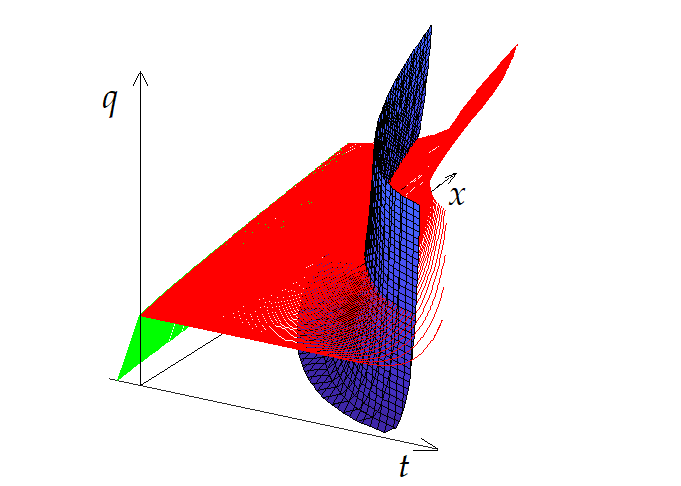}
\par\end{centering}
\vspace{-0.2in}
 \hspace*{-0.1in}\caption{\label{PowSwSingExt} Switching surface (blue), singular surface (red), and safe switching surface (green) for $P(x)=x^{0.7}$ and $L(y)=2y^{0.5}$, shown from two different angles.}
\end{figure}

To find the ``safe" switching time $t_m$, note that $x$ is constant over an R arc, and $u=0$, so $\dot{\lambda} = -\mu P'(x)=-\mu(T_0) P'(x_m)$, where $x_m:=x(t_m)$. Therefore, on $[t_m, T_{0}]$,
$$
\lambda(t)=\lambda(T_0)+\mu(T_0)P'(x_m)(T-t). 
$$
Since $\lambda(t)$ equals $\mu(T_0)$ at $t=t_m$ we derive:
\begin{equation}\label{SafeSwt}
t_m=T_0-\frac1{P'(x_m)}\left(1-\frac{\lambda(T_0)}{\mu(T_0)}\right). 
\end{equation}
Together with $t_m=\int\limits_{x_0}^{x_m}\frac{dx}{P(x)}$, which integrates the state equation over the preceding V arc, this determines $t_m$ and $x_m$. The time $t_m$ is the time when the plant is old enough to start reproducing, and is called the {\it age of maturity}. We can show that if the initial mass of the plant is high enough and the safe period is long enough, the age of maturity is always positive, even if the plant has no reproductive mass initially (as is usually the case). If the safe period is sufficiently long, it is better to invest into growth first, to produce more offspring later.

\begin{corollary} If $T_0P'(x(T_0))\geq1$, then the initial arc of an optimal trajectory is always V. In particular, if $P(x) = x$, and $T_0>1$ then the initial arc is V for any initial values. 
\end{corollary}
\begin{proof}
Since the marginal values are positive,  $x(T_0)=x_m$, and the growth is purely reproductive after $t_m$, we must have $t_m>T_0-\frac1{P'(x_m)}=T_0-\frac1{P'(x(T_0))}>0$ by \eqref{SafeSwt}. \end{proof}
\noindent The value of $x(T_0)$ is typically easy to estimate from above in terms of $x_0$ and $T_0$ by integrating $\dot{x}=P(x)$ on $[0,T_0]$. When $P(x)=kx^\beta$ we have explicitly $x(T_0)=\left(x_0^{1-\beta}+(1-\beta)kT_0\right)^{\frac1{1-\beta}}$.

As before, when $P(x) = x$ and $L(y) = \frac{y^\alpha}{\alpha}$, the behavior of optimal trajectories can be represented by a two-dimensional mass ratio diagram for $q=\frac{y}{x}$, Figure \ref{Extdiagram}. The maturity surface is represented by a curve that starts from the $t$ axis and intersects the $q$ axis at the same point as the singular curve $q_{\textrm{sing}}=\frac{\alpha}{\exp(\frac{\alpha}{1-\alpha}\left(T-\tau_{si})\right)}$. We denote its $t$-intercept $T_0-\tau_R$. 
\begin{figure}[htbp]
\begin{centering}
\includegraphics[width=150mm]{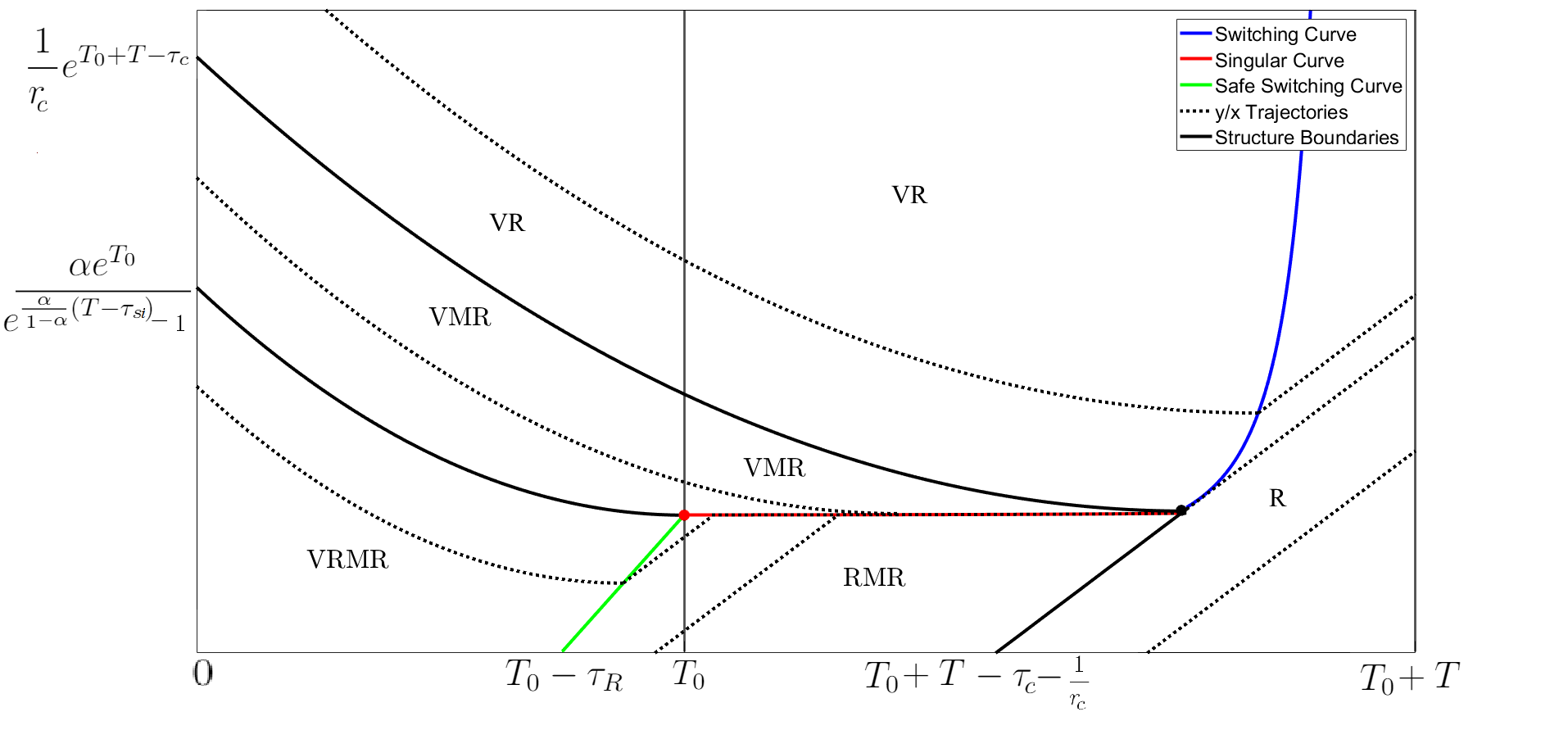}
\par\end{centering}
\vspace{-0.4in}
\hspace*{-0.1in}\caption{\label{Extdiagram} The mass ratio diagram for the extended time model with $P(x)=x$ and $L(y) = \frac{y^\alpha}{\alpha}$ with $\alpha<1$. The inactive parts of the switching and singular curves are dashed, the vertical lines are their asymptotes. V arcs fall exponentially, R arcs rise with slope $1$, M arcs follow the red curve. Optimal trajectories starting in the labeled regions have the indicated control structure.}
\end{figure}

Finally, we state an analog of Theorem \ref{4LinStructures} for 
the case of $T>q_c+\tau_c$ and $T_0>\tau_R$, for which the proof is analogous. The modifications for shorter times, when pure R and RMR trajectories also occur, are straightforward. 
\begin{theorem}\label{Ext3LinStructures} Let $P(x) = x, L(y) = \frac{y^\alpha}{\alpha}$, and $q_0=\frac{y_0}{x_0}$ be the ratio of initial values. Then:

\textup{(i)} If $q_0<\frac{\alpha e^{T_0}}{\exp(\frac{\alpha}{1-\alpha}\left(T-\tau_{si})\right)-1}$, then the optimal trajectory is VRMR.
\smallskip 

\textup{(ii)} If $\frac{\alpha e^{T_0}}{\exp(\frac{\alpha}{1-\alpha}\left(T-\tau_{si})\right)-1}
\leq q_0<\frac{1}{r_c}e^{T_0+T-\tau_c}$, then the optimal trajectory is VMR.
\smallskip 

\textup{(iii)} If $q_0\geq\frac{1}{r_c}e^{T_0+T-\tau_c}$, then the optimal trajectory is VR.
\end{theorem}


\bigskip
\noindent{\large\bf Acknowledgments:} This work was conceived during the summer 2019 REU program at the University of Houston-Downtown, and is funded by the National Science Foundation grant 1560401. The authors are grateful to the anonymous reviewers for clarifying the relation between density dependence and fitness and many other helpful suggestions that greatly improved the paper.

\end{document}